\documentclass[11pt,envcountsect]{style}
\usepackage[margin=1in]{geometry}
\usepackage[final]{graphicx}
\usepackage{color}
\usepackage{csquotes}

\usepackage{amsmath, amssymb}
\usepackage{graphicx}
\usepackage{latexsym,amsfonts,epsf}
\usepackage{hyperref}

\usepackage{microtype}
\usepackage{paralist}
\usepackage{aliascnt}

\newcommand {\ignore} [1] {}

\usepackage{algorithm}
\usepackage[noend]{algpseudocode}

\def\Prob{\mathbb{P}\mathrm{r}}
\def\Var{\mathbb{V}\mathrm{ar}}
\def\Exp{\mathbb{E}}
\def\eps{\varepsilon}
\def\reals{\mathbb{R}}
\def\nats{\mathbb{N}}
\def\D{\mathcal{D}}

\def\B{B}

\def\p{\hat{p}}
\def\q{\hat{q}}

\usepackage{bm}
\usepackage{mathtools}

\newcommand{\lp}{\left}
  \newcommand{\rp}{\right}
\newcommand{\mrm}[1]{\mathrm{#1}}
\newcommand{\mcal}[1]{\mathcal{#1}}
\newcommand{\mfrak}[1]{\mathfrak{#1}}
\DeclareMathOperator*{\argmin}{argmin}

\newcommand{\bigOh}[1]{O\lp(#1\rp)}

\newcommand{\bigOmega}[1]{\Omega\lp(#1\rp)}

\newcommand{\N}{\mathbb{N}}

\newcommand{\me}{\mathrm{e}}

\newcommand{\Normal}[2]{\mcal{N}\lp(#1,#2\rp)}
\newcommand{\DNormal}[2]{\mrm{DN}\lp(#1,#2\rp)}
\newcommand{\NormalPDF}[3]
{
  \frac{1}{\sqrt{2 \pi} #2} \me^{-\frac{(#3 - #1)^2}{2 {#2}^2}}
}
\newcommand{\erf}[1]{\mrm{erf}\lp(#1\rp)}
\newcommand{\erfc}[1]{\mrm{erfc}\lp(#1\rp)}

\newcommand{\vProb}[1]{\mathbb{P}\mathrm{r}\lp[#1\rp]}
\newcommand{\vVar}[1]{\mathbb{V}\mathrm{ar}\lp[#1\rp]}

\newcommand{\dkl}[2]{D_{\mrm{kl}} \lp(#1\|#2\rp) }
\newcommand{\dtv}[2]{d_{\mrm{tv}} \lp(#1, #2\rp) }
\newcommand{\dhel}[2]{d_{\mrm{hel}} \lp(#1, #2\rp) }

\newcommand{\err}[1]{\mbox{\rm err}\left(#1\right)}

\begin{document}

\title{Learning Powers of Poisson Binomial Distributions}

\author{
  Dimitris Fotakis \inst{1, 2}
  \and
  Vasilis Kontonis \inst{1}
  \and
  Piotr Krysta \inst{3}
  \and
  Paul Spirakis \inst{3, 4}
}
\institute{ National Technical University of Athens,
  School of Electrical and Computer Engineering,
  Iroon Polytechneiou 9, 15780 Athens, Greece.
  Email: {\tt fotakis@cs.ntua.gr, vkonton@gmail.com}
  \and
  Yahoo Research, 229 West 23rd Street, 10036 NY.
  Email: {\tt dfotakis@yahoo-inc.com}
  \and
  University of Liverpool,
  Department of Computer Science,
  Ashton Building, Ashton Street Liverpool L69 3BX, U.K.
  Email: {\tt p.krysta@liverpool.ac.uk, spirakis@liverpool.ac.uk}
  \and
  CTI, Greece.
  Email: {\tt spirakis@cti.gr}
}

\ignore{
  \author{
    Dimitris Fotakis\thanks{National Technical University of Athens,
      School of Electrical and Computer Engineering,
      Iroon Polytechneiou 9, 15780 Athens, Greece.
      Email: {\tt fotakis@cs.ntua.gr}.}
    \and
    Vasilis Kontonis\thanks{National Technical University of Athens,
      School of Electrical and Computer Engineering,
      Iroon Polytechneiou 9, 15780 Athens, Greece.
      Email: {\tt vkonton@gmail.com}.}
    \and
    Piotr Krysta\thanks{University of Liverpool,
      Department of Computer Science,
      Ashton Building, Ashton Street Liverpool L69 3BX U.K.
      Email: {\tt p.krysta@liverpool.ac.uk}.}
    \and
    Paul Spirakis
    \thanks{University of Liverpool,
      Department of Computer Science,
      Ashton Building, Ashton Street Liverpool L69 3BX U.K.
      Email: {\tt spirakis@liverpool.ac.uk}.}
  }
}

\maketitle

\vspace*{-6mm}

\begin{abstract}
  We introduce the problem of simultaneously learning \emph{all powers} of
a Poisson Binomial Distribution (PBD). A PBD over $\{1,\ldots,n\}$ is the
distribution of a sum $X = \sum_{i=1}^n X_i$, of $n$ independent Bernoulli
$0/1$ random variables $X_i$, where $\Exp[X_i] = p_i$. The $k$'th \emph{power}
of this distribution, for $k$ in a range $\{1,\ldots, m\}$, is the distribution
of $P_k = \sum_{i=1}^n X_i^{(k)}$, where each Bernoulli random variable
$X_i^{(k)} \in \{0,1\}$ has $\Exp[X_i^{(k)}] = (p_i)^k$. The learning algorithm
can query any power $P_k$ several times and succeeds in simultaneously learning
all powers in the range, if with probability at least $1- \delta$: given
any $k \in \{1,\ldots, m\}$, it returns a probability distribution $Q_k$
with total variation distance from $P_k$ at most $\varepsilon$.

We provide almost matching upper and lower bounds on the query complexity
for this problem. We first show an information theoretic lower bound on
the query complexity on PBD powers instances with many distinct parameters
$p_i$ which are significantly separated. This lower bound shows that essentially
a constant number of queries is required per each distinct parameter.
We almost match this lower bound by examining the query complexity of
simultaneously learning all the powers of a special class of PBD's resembling
the PBD's of our lower bound.

We extend the classical minimax risk definition from statistics,
dating back to 1930s [Wald 1939], to introduce a framework to study
sample complexity of estimating functions of sequences of distributions.
Within this framework we show how classic lower bounding techniques,
such as Le Cam's and Fano's, can be applied to provide sharp lower bounds in
our learning model.

We study the most fundamental setting of a Binomial distribution, i.e.,
$p_i = p$, for all $i$, and provide an optimal algorithm which uses
$O(1/\eps^2)$ samples, independent of $n, m$. Thus, we show how to
exploit the additional power of sampling from other powers, that leads to
a dramatic increase in efficiency. We also prove a matching lower bound of
$\Omega(1/\varepsilon^2)$ samples for the Binomial powers problem, by
employing our minimax framework.

Estimating the parameters of a PBD is known to be hard. Diakonikolas, Kane and
Stewart [COLT'16] showed an exponential lower bound of $\Omega(2^{1/\eps})$ samples
to learn the $p_i$'s within error $\eps$. Thus, a natural question is whether sampling
from powers of PBDs can reduce this sampling complexity. Using our minimax framework
we provide a negative answer to this question, showing that the exponential number
of samples is inevitable. We then give a nearly optimal algorithm that learns the
$p_i$'s of a PBD using $2^{\bigOh{n \max(\log(1/\eps), \log(n))}}$ samples from
the powers of the PBD, which almost matches our lower bound.

The Newton-Girard formulae give relations between the power sums
$\sum_{i=1}^n z_i^k$, $k = 1, \ldots, n$, of the roots, and the
coefficients of a polynomial $P(x)= \prod_{i=1}^n (x-z_i)$. Thus, if we
know the power sums exactly, we can first find the coefficients of $P(x)$
and then compute the roots $z_1, \ldots, z_n$ with an arbitrarily
good accuracy. In our problem we only have access to approximate
values for the power sums since they correspond to the means of the PBD powers.
An intriguing question is to which extent these \enquote{noisy} power sum
estimations can be used to recover the actual values of $p_1, \ldots, p_n$
within sufficient accuracy. We answer this question by providing close
lower and upper bounds on the sample complexity of estimating the parameters
of a PBD using samples from its powers.

\end{abstract}

\setcounter{page}{0}
\thispagestyle{empty}
\newpage

\pagestyle{plain}

\section{Introduction}
\label{intro}

\subsection{Our Model and the PBD Powers Problem}
A Poisson Binomial Distribution (PBD) is the discrete probability
distribution of a sum of $n$ independent Bernoulli indicator random variables,
and $n$ is the \emph{order} of the distribution. So if $X$ is a PBD of
order $n$ then $X = \sum_{i=1}^n X_i$ where $X_1,...,X_n$ are independent
Bernoulli $0/1$ random variables. The expectations $(\Exp[X_i] = p_i)_i$, called the
\emph{parameters of the PBD}, do not need to be the same and thus PBD's capture a
quite wide class of distributions. It is believed that Poisson was the first to
consider PBD's, hence their name. Let now a random variable $Y_{i,k}$ be the
product of $k$ Bernoulli independent random variables, each distributed as $X_i$.
The expectation of $Y_{i,k}$ is $(p_i)^k$. If $P_k$ is the sum $\sum_{i=1}^n
Y_{i,k}$, we call the PBD $P_k$ the $k$th \emph{power} of the PBD $X$. The
expectation of
$P_k$ is equal to $\sum_{i=1}^n (p_i)^k$. The powers of a PBD clearly relate to
the moments of  the PBD.

Suppose an unsupervised learning algorithm knows $n$ but not the
$p_i$'s, and aims at approximately and simultaneously learning \emph{all} the
powers $P_k$ of a PBD $X$ for $k \in \{1,\ldots, m\}$, where $m$ is given and
can even be greater than $n$. The algorithm can ask for independent samples
from any $P_k$ for $k$ in any subset of the {\em range} $\{1,\ldots, m\}$.
A query to $P_k$ returns an independent sample from distribution $P_k$.
Each such sample has $\log n$ bits since by definition the maximum value
of $P_k$ is $n$. The algorithm can proceed in an adaptive way, by getting
some samples from some powers, then computing, then asking for more
samples, depending on the computations and previous samples.
The algorithm is said to \emph{succeed} with probability at least $1-\delta$, for given
$\delta > 0$, if the following holds with probability at least $1-\delta$:
Given any $k \in \{1, \ldots, m\}$, the algorithm outputs a distribution $Q_k$
whose Total Variation Distance from $P_k$ is at most $\eps$. Here,
$\delta, \eps \in (0,1)$ are given as input. Note, $Q_k$ is not
needed to be a PBD itself. The \emph{query complexity} of the algorithm is the
total number of samples obtained and is a function of $n$, $m$,  $1/\delta$,
$1/\eps$. To compare different algorithms that query for independent
samples from a subset in the range and manage to learn \emph{all}
powers in the range, we consider the query complexity per learned
power to be the total number of queries divided by the number of
powers we learn. We study this problem of simultaneously learning all the
powers of a PBD in a given range in terms of query and time complexity
efficiency. Ideally, our learning algorithm runs in time \emph{polynomial}
in $n$, $m$, $1/\delta$, $1/\eps$, but our primary focus is query complexity.
The problem can of course be solved by taking samples per power to learn it
approximately for each power in the range. The challenging question is if we
can do much better than this in terms of query and time complexity, given
the fact that the powers of the unknown PBD are related because they are
defined over the same unknown parameters $p_i$'s.

\subsection{Motivation}
\subsubsection{Random Coverage Valuations}
The PBD powers problem arises from the problem of learning a natural class of
\emph{random coverage valuations}. Given a ground set $X = \{ e_1, \ldots, e_n
\}$, a function $v : 2^U \to \nats$ is a \emph{coverage valuation} if there are
$A_1, \ldots, A_m \subseteq X$ such that for all $S \subseteq [m]$, $v(S) =
\left| \bigcup_{j \in S} A_j \right|$.
Coverage valuations are monotone and submodular and have received considerable
attention in optimization (maximizing a coverage valuation under a cardinality
constraint), learning and algorithmic mechanism design, see e.g.,
\cite{BH11,DS16,FK14} and the references therein.

Let now each element $e_i \in X$ be associated with a
probability $p_i \in [0, 1]$ and we generate $m$ random subsets $A_1, \ldots,
A_m \subseteq X$ by including each $e_i \in X$ in each $A_j$ independently
with probability $p_i$. The random sets $A_1, \ldots, A_m$ are selected
independently and are identically distributed. Random sets $A_1, \ldots, A_m$
define a random coverage valuation function $v : 2^{[m]} \to \nats$ with
$v(S) = \left| \bigcup_{j \in S} A_j \right|$.
      Suppose we are interested in approximately learning the distribution of
the values of such random coverage valuations $v$ evaluated over subsets $S
\subseteq [m]$. Namely, given a ground set $X = \{ e_1, \ldots, e_n \}$ and the
probabilities $p_1, \ldots, p_n$, we want to find a family of probability
distributions $\D(S)$ so that
$ \Prob[\D(S) = i] \approx \Prob\left[ \left| \bigcup_{j \in S} A_j \right| = i
\right]$,
for all $i \in \{ 0, \ldots, n \}$ and $S \subseteq [m]$ (the probability in the
right-hand-side is taken over the random sets $A_j$ with $j \in S$). Each
$\D(S)$ approximates the distribution of the values $v(S)$ of a coverage
valuation function $v$ chosen randomly from the family of coverage
valuations described above.
    Since the random sets $A_1, \ldots, A_m$ are independently identically
distributed, only the cardinality of $S$, and not $S$ itself, matters for the
union's cardinality. Hence, given $X$ and the probabilities $p_i$'s, we aim to
compute probability distributions $\D_k$ so that
$\Prob[\D_k = i] \approx \Prob\left[ \left| \bigcup_{j=1}^k A_j \right| = i
\right]$,
for all $i \in \{ 0, \ldots, n \}$ and $k \in \nats$. Each $\D_k$ approximates
the distribution of the cardinality of the union of $k$ sets selected randomly
and independently from $X$.

Learning random coverage valuations can be reduced to the PBD powers problem,
by observing that each distribution $\D_k$ is the sum of $n$ independent Bernoulli
variables with expectations $1-(1-p_1)^k, \ldots, 1-(1-p_n)^k$, where each such
Bernoulli variable $i$ indicates whether element $e_i$ is included in at least one
of the $k$ random sets considered in the union. A natural sampling model is that
the learning algorithm selects an index $k \in \nats$ and receives the cardinality
of the union of $k$ random sets, which is exactly the sampling model in the PBD
powers problem.

\subsubsection{Newton's identities}
The Newton-Girard formulae give relations between the power sums
$\sum_{i=1}^n z_i^k$, $k = 1, \ldots, n$, of the roots, and the
coefficients of a polynomial $P(x)= \prod_{i=1}^n (x-z_i)$. Thus, if we
know the power sums exactly, we can first find the coefficients of $P(x)$
and then compute the roots $z_1, \ldots, z_n$ with an arbitrarily
good accuracy. A similar approach was used in \cite{DP13} to derive sparse
covers for PBDs. In our problem we only have access to approximate
values for the power sums since they correspond to the means of the PBD powers.
An intriguing question is to which extent these \enquote{noisy} power sum
estimations can be used to recover the actual values of $p_1, \ldots, p_n$
within sufficient accuracy. We answer this question by providing close
lower and upper bounds on the sample complexity of estimating the parameters
of a PBD using samples from its powers.

\subsection{Our Results}
We now state our first lower bound. A vector $\bm{p} = (p_1, \ldots, p_n) \in [0, 1]^n$ is called
\emph{$(\nu, \kappa, m)$-separated}, for some positive integers $m$ and $\kappa
> \nu$, with $n/m$ also a positive integer, if there are $m$ positive integers
$a_1, \ldots, a_m \in [\nu]$ so that $\bm{p}$ contains a group of $n/m$ values
$p_i = 1 - a_i / \kappa^i$ for each $i \in [m]$. Thus, a
$(\nu, \kappa, m)$-separated vector $\bm{p}$ has $n/m$ entries of value
$p_1 = 1 - a_1/\kappa$, $n/m$ entries of value $p_2 = 1
- a_2 / \kappa^2$, \ldots, and $n/m$ entries of value $p_m = 1 - a_m /
\kappa^m$. A PBD $X$ is \emph{$(\nu, \kappa, m)$-separated} if the
parameters defining $X$ are given by a $(\nu, \kappa, m)$-separated vector $\bm{p}$.

Our results indicate that when the separation of the $p_i$'s is substantial
the problem of estimating the densities of the PBD powers is equivalent to
approximate the PBD's parameters. The following information-theoretic lower bound
shows that for any integer $m \leq n / (\ln n)^4$, learning an appropriate
collection of $m$ powers of an $(\ln n, (\ln n)^4, m)$-separated PBD $\bm{p}$
requires $\Omega(m \ln\ln n/\ln n)$ samples in the worst case.
Hence, for the special case of separated PBDs, the sampling complexity should increase
almost linearly with the number $m$ of different $p_i$ values in $\bm{p}$,
at least as long as $m \leq n / (\ln n)^4$.

\begin{theorem}\label{th:lower_bound}
  For any positive integer $m \leq n / (\ln n)^4$ so that $n/m$ is an integer,
  any algorithm that succeeds in learning all powers with indices $(\ln
  n)^{4i-2}$, $i = 1, \ldots, m$, of an $(\ln n, (\ln n)^4, m)$-separated PBD,
  within total variation distance $\eps \in (0, 1/4]$ and with failure
  probability $\delta \leq 1/2$, requires $\Omega(m \ln\ln n/\ln n)$ samples in
  the worst case.
\end{theorem}

To almost match this lower bound, we show how to learn the following
class of PBD's resembling our lower bound PBD. Let
$p_i = 1 - \alpha_i/(c \cdot \ln(n))^{s-i}$, with $c >1$ any constant,
and $s$ a number such that $(c \cdot \ln(n))^{s} = n$,
$i=0,1,\ldots,s-1$. Notice, $s \approx \ln(n)/\ln(\ln(n))$, and assume
$\alpha_i \in \{1,2,\ldots, \lfloor \sqrt{\ln(n)} \rfloor \}$ for each $i$.
   The class $\cal P$ of PBD's instances has $n_i$ probabilities
equal to $p_i$ for $i=0,1,\ldots,s-1$, where $n_i = n/s$ for each
$i=0,1,\ldots,s-1$. We assume that $n$ and $s$ are known. The mean of the
first power of a PBD from $\cal P$ is $\sum_{i=0}^{s-1} n_i p_i$.
This PBD is defined by a $(\sqrt{\ln(n)}, c \cdot \ln(n), s)$-separated vector
$\bm{p}$. We call a {\em block} $i$, all the $n_i$ probabilities equal to
$p_i$, and note that $p_0 > p_1 > \cdots > p_{s-1}$.

\begin{theorem}\label{th:separated_upper}
  Let $c \geq 2$ be a constant, $\varepsilon \leq 1/(6c)$, and $n \geq
  \max\{e^{2c}, 4/((2-\sqrt{2})^2\varepsilon^2)\}$. Given an unknown PBD
  $X \in {\cal P}$, the exact values of $\alpha_i$ for each $i=0,1,\ldots,s-1$
  can be learned by Algorithm \ref{alg:PBD_class} using
  $\bigOh{\log(s/\delta)/\varepsilon^2}$ samples from each power
$X^{\ell_i}$, where $\ell_i = (c \cdot \ln(n))^{s-i}/c$ for $i=0,1,\ldots,s-1$, with
  success probability at least $1-\delta$. The total number of samples is
$\bigOh{s\log(s/\delta)/\varepsilon^2} $.
\end{theorem}

Although our algorithm finds exact values of $p_i$'s,
thus learning all the powers, it uses at most
$O(\ln(\ln(n)/(\delta \ln\ln(n)))/\eps^2)$
samples per sampled power, which is very close to our lower bound.
Interestingly, our lower bound proof shows that the claimed number of
samples recover the exact values of $p_i$'s. Thus, it essentially shows
that $\Omega(1)$ samples are required per distinct value of $p_i$, and
our upper bound uses roughly $\bigOh{\ln(\ln (n))}$ samples per distinct
value of $p_i$.

The lower bound of Theorem \ref{th:lower_bound} implies that the problem is hard in general,
which motivates us to consider the PBD powers problem with few distinct parameters. We ask: does
the additional power of the algorithm of being able to sample from many powers
make the parameter estimation easier? We answer this question in the negative,
by proving an exponential lower bound in this case.

The classic minimax risk framework from statistics is used for investigating lower
and upper bounds on the sample complexity of testing and learning a single distribution,
cf.~\cite[Chapter 2]{tsybakov_2008}. We generalise and extend this framework from
testing and learning a single distribution to sequences of distributions.
This generalisation is new to the best of our knowledge.
   The main ingredients of our framework are generic theorems that reduce
the problem of learning a sequence of distributions to testing such
sequence and provide generic lower bounds on the minimax risk
based on classical results from statistics, see, e.g.,
\cite{wald_1939,yu_1997,birge_1983,yang_barron_1999} and \cite[Chapter 2]{tsybakov_2008}.
For precise formulations of these new theorems, see Proposition~\ref{pr:estimation_testing},
Lemma~\ref{l:minimax_lecam_lower_bound}  and
Lemma~\ref{l:minimax_fano_lower_bound}.

Crucial to our framework and our main conceptual contribution here,
is the new definition, Definition \ref{def:adaptive_minimax_risk}, of the minimax
risk for sequences of distributions. This definition unveils the structure of any
learning algorithm in our model. Namely, such algorithm has two distinct
types of operations related to sampling, that of deciding from which distributions
in the sequence to sample, and that of using the samples in its computation phase.
The two operations' types might alternate and the algorithm may be adaptive or
non-adaptive about further sampling and using the samples.

Our framework is of independent interest since it can be used for proving
lower bounds for estimation of functions of distribution sequences such as
their densities or their parameters.  It can be instantiated with the
classic methods for proving sample lower bounds, that is, Le Cam's, Fano's
or Assouad's methods \cite{yu_1997}. We present two applications of
this framework to prove two lower bounds, in Theorems \ref{th:lower_bound_exp} and
\ref{th:binomial_lower_bound}. To prove an exponential lower bound for
parameter estimation in our model, we use our framework with the
Le Cam's method and a PBD instance introduced in Proposition 15
of \cite{DKS16colt2}.
\begin{theorem}\label{th:lower_bound_exp}
  If $n \geq 1/\eps$, then any learning algorithm that draws $N$ samples from
any powers of an PBD of order $n$ and returns estimates of the parameters of
this PBD within additive error $\eps$ with success probability at least $2/3$
must have $N \geq 2^{\Omega(1/\eps)}$.
\end{theorem}

We see that parameter estimation remains very difficult even if we enrich the
power of the algorithm to allow for sampling from any power of the PBD with $n
= \Theta(1/\eps)$. In sharp contrast, observe, that using the density
estimation algorithms from \cite{DDS11,DKS16colt1} for each of the $n$ powers
of this PBD, we can learn the densities of all these powers with only
$\tilde{O}(n/\eps^2) = \tilde{O}(1/\eps^3)$ samples. This gives a provable
separation in the sampling complexity between parameter and density estimation
in our model, even if the PBD has a constant number of distinct parameters.
This also implies that we cannot use parameter estimation in our model as means
for density estimation if the underlying PBD has a constant number of distinct
parameters.

We almost match the exponential lower bound of Theorem \ref{th:lower_bound_exp}
for parameter estimation with a close upper bound in the most general version of the
PBD powers model. We use the Newton-Girard
identities to reduce our problem to the classical polynomial root finding problem.
We present an analysis of the error of this reduction from power sums to
coefficients of the polynomial and then to its roots, when power sums are known
only approximately. The main obstacle in this approach is that to find the roots
of a polynomial with inexact coefficients we need extremely good approximations
of the coefficients and this leads to an exponential number of samples. Since the
algorithms for root finding are almost optimal, this exponential upper bound
cannot be improved, unless one uses a different technique. This leads to the
following result with details in Appendix \ref{s:app:newton_identities_samples}.

\begin{theorem}\label{th:newton_identities_samples}
  Let $X$ be a PBD with probability vector $\bm{p}$. There exists an algorithm
  which draws $2^{\bigOh{n \max \lp(\log(1/\eps), \log n\rp)}}$ samples from the
  powers of $X$ and computes a vector $\hat{\bm{p}}$ such that
  $ \| \bm{p} - \hat{\bm{p}} \|_{\infty} \leq \eps.$
\end{theorem}

Given the two lower bounds we turn our attention to investigating the
model with few distinct paramaters, focusing on a single
parameter, i.e., the Binomial case. Here, we prove that the parameter
and density estimations are essentially equivalent. That is, we get a dramatic
increase in efficiency and design an elegant algorithm which learns {\em all}
powers of a given Binomial using only a constant $O(1/\eps^2)$ number of
samples. Crucial for our solution is to generalise the PBD powers
problem allowing for any positive real powers. Below $\B(n,p)$ is the
PBD with all parameters equal to $p$. (Algorithm~\ref{alg:binomial} can be
generalised to allow $p \in [\eps^2/n^d, 1-\eps^2/n^d]$, for any constant $d$,
see Section~\ref{s:app:binomial:hugep}.)

\begin{theorem}\label{th:binomial}
  Let $\eps \in (0, 1/6), n \in \nats$.
  Then, for any $p \in [\eps^2/n, 1-\eps^2/n]$, Algorithm~\ref{alg:binomial}
uses
  $O(\ln(1/\delta)^2/\eps^{2})$ samples and outputs $\hat{a}\in \reals_{++}$,
  $\hat{q}_1, \hat{q}_2 \in (0,1)$ such that
  $\dtv{B(n, \hat{q}_1^l)}{B(n,p^{\hat{a}l})} = O(\eps)$
  for $l \in (1,+\infty)$ and $\dtv{B(n,\hat{q}_2^l)}{B(n,p^{\hat{a}l})} =
O(\eps)$
  for $l \in(0,1)$ with probability at least $1 - \delta$.
\end{theorem}

It\rq s well known that to distinguish two given Binomial distributions,
$\Omega(1/\eps^2)$ samples are required, e.g., \cite{DDS11}, implying
the same lower bound for learning a single Binomial. This lower bound does
not apply to our model, because in our setting the input of the algorithm
contains samples from many different distributions.
To prove a matching lower bound we use our minimax framework with Fano's
method.

\begin{theorem}\label{th:binomial_lower_bound}
  Let $A$ be an algorithm that returns probability distributions which
  are within total variation distance $\eps$ from $B(n,p^i)$ for all
  $i \in \{1,2,3, \ldots \}$, using samples from the distributions $B(n,p^i)$
  with probability of success at least $2/3$. Then $A$ uses
  $\bigOmega{1/\eps^2}$ samples.
\end{theorem}

\subsection{Related Work}
The problem of approximately learning a PBD, within a given total variation
distance $\eps$, in a sample and time efficient way, is a fundamental problem
in unsupervised learning and has received significant attention. Chebyshev's
inequality gives an optimal bound of $O(1/\eps^2)$ on the number of samples for
learning a Binomial distribution of known order $n$ with constant failure
probability. Birg\'e \cite{Bir97} gave an efficient algorithm for learning any
continuous unimodal distribution over $\{0, \ldots, n\}$ with $O(\log
n/\eps^3)$ samples (this result can be extended to PBDs \cite{DDS11}), and
proved that this sample complexity is essentially best possible for unimodal
distributions. By an elegant combinatorial construction, Daskalakis and
Papadimitriou \cite{DP13} proved that the family of all PBDs admits a
\emph{sparse cover}, i.e., there is a small subset of PBDs, of size $n^2 + n
(1/\eps)^{O(\log^2(1/\eps))}$, so that every PBD is within a total variation
distance of $\eps$ to some PBD in the subset. They used the sparse cover of
PBDs to efficiently compute an approximate Nash equilibrium in anonymous
multiplayer games \cite{DP15}. Daskalakis, Diakonikolas and Servedio
\cite{DDS11} exploited the sparse cover of PBDs (and several other ideas and
techniques) to show that a PBD can be learned approximately with
$O(\ln(1/\delta)/\eps^2)$ samples, where $\delta$ is the probability of
failure, and in $O((1/\eps)^{\mathrm{poly}(\log(1/\eps))} \log n)$ time. Other
applications of the sparse cover of PBDs include efficient testing of whether a
given distribution is $\eps$-close to some PBD \cite{AD15}. The result and the
techniques of \cite{DDS11} were generalised to learning sums of independent
integer random variables \cite{DDS13} and to learning Poisson Multinomial
Distributions (PMDs) \cite{DKT15}. Very recently, efficient algorithms have
been shown for learning PMDs \cite{DDKT16,DKS16}, PBDs \cite{DKS16colt2}, and
sums of independent integer random variables \cite{DKS16colt1}, by using
Fourier Transforms. The paper \cite{DKS16colt1} implies that PBDs can be
learned approximately with $O((1/\eps^2) \sqrt{\log(1/\eps)})$ samples and with
the same running time.

In this very active research area (see also \cite{Diak16} for a survey and
references on efficient approximate learning of structured probability
distributions), we consider the problem of approximately learning, in a sample
and time efficient way, a family of many closely related PBDs (instead of a
single one). Given some samples, we need to extract information not only about
the parameters of the PBD from which the samples come, but also about the
relation of the different probability distributions that they generate the
samples. To the best of our knowledge, this is the first time that a similar
question has been studied in the area of unsupervised learning of structured
probability distributions.

\subsection{Notation}
Our model and the basic definitions are introduced at the beginning of
Section~\ref{intro}. We introduce here some additional notation used throughout
the paper. Additional notation will be introduced per Section. For any positive
integer $k$, we let $[k] = \{1, \ldots, k\}$.  We let $\log n$ be the base-$2$
logarithm of $n$ and let $\ln n$ be the natural
logarithm of $n$. We let $\Exp[X]$ and $\Var[X]$ denote the mean value and the
variance, respectively, of a probability distribution $X$. We let $\B(n, p)$
denote a binomial distribution of order $n$ and probability $p$. We usually
identify a PBD with the vector $\bm{p} = (p_1, \ldots, p_n)$ of expectations
of its Bernoulli trials.
We denote $\err{n, p, \eps} = \eps \sqrt{p(1-p)/n}$.
Let $\reals_{++} = \{a \in \reals\, :\, a>0\}$.
Given two probability distributions $X$ and $Y$ over the set of natural numbers
$\nats$, the \emph{total variation distance} (TVD) of $X$ and $Y$, denoted
by $\dtv{X}{Y}$ is $\dtv{X}{Y} = \textstyle \sum_{i \in \nats} |\Prob[X = i] - \Prob[Y = i]|/2$.
For brevity, we often refer to the total variation distance simply
the distance of $X$ and $Y$ or TVD of $X$ and $Y$.  We often use $X(i)$ to
denote $\Prob[X = i]$, i.e., the probability that $X$ takes the value $i$.

\section{Lower Bound for Learning PBD Powers: Separated Case} \label{s:lower_bound}
\subsection{Preliminaries}
We show a simple lower bound on the total variation distance of two PBDs
based on the difference of their expected values. Its proof can be found
in Section~\ref{s:app:tvd_lb} of the Appendix.
\begin{lemma}\label{l:tvd_lb}
  Let $X$ and $Y$ be two PBDs with expected values $\mu_X = \Exp[X]$ and $\mu_Y =
  \Exp[Y]$ and variances $\sigma^2_X = \Var[X]$ and $\sigma^2_Y = \Var[Y]$. Then,
  for any $\eps > 0$ such that $\sigma^2_X, \sigma^2_Y \geq
  \ln(\frac{2}{1-\eps})$, if $|\mu_Y - \mu_X| >
  2\sqrt{\ln(\frac{2}{1-\eps})}(\sigma_X + \sigma_Y)$, then $\dtv{X}{Y} > \eps$.
\end{lemma}

\subsection{The Proof of Theorem~\ref{th:lower_bound}}
We show that learning $m$ appropriately selected powers of an $(\ln n, (\ln n)^4,
m)$-separated PBD $\bm{p}$ requires $\Omega(m \ln\ln n/\ln n)$ samples, i.e.,
almost as many as the number of different $p_i$ values in $\bm{p}$,
provided $m \leq n / (\ln n)^4$. We assume that $n$ is large enough and
consider an $(\ln n, (\ln n)^4, m)$-separated PBD defined by
$m \leq n/(\ln n)^4$ integers $a_1, \ldots, a_m$, with $1 \leq a_i \leq \ln n$.
The corresponding vector $\bm{p}$ consists of $m$ groups with $n/m$ entries
$p_i = 1 - a_i / (\ln n)^{4i}$ in each group
$i = 1, \ldots, m$ (note, $p_1 < p_2 < \cdots < p_m$).
For simplicity, we assume that $\ln n, n / m \in \nats$.
The intuition behind the proof is that given a distribution that approximates
the $(\ln n)^{4i-2}$-th power of $\bm{p}$, we can extract the exact value of
$a_i$. Lemma \ref{l:lb} helps towards formalizing this intuition.

\begin{lemma}\label{l:lb}
For any $m \in \nats$, $m \leq n / (\ln n)^4$, let $\bm{p}$ and $\bm{q}$ be two
$(\ln n, (\ln n)^4, m)$-separated vectors defined by positive integers $a_1,
\ldots, a_m$ and $b_1, \ldots, b_m$, resp. For any $\eps \in (0, 1/2]$
and any $i \in [m]$, if the $(\ln n)^{4i-2}$-th power of $\bm{p}$ and the
$(\ln n)^{4i-2}$-th power of $\bm{q}$ are within TVD at most $\eps$, then $a_i
= b_i$.
\end{lemma}

\begin{proof} (sketch)
For contradiction, assume that there is $i \in [m]$ so that the
$(\ln n)^{4i-2}$-th power of $\bm{p}$ and the $(\ln n)^{4i-2}$-th
power of $\bm{q}$ are within TVD at most $\eps$ and $a_i < b_i$ (case
$a_i > b_i$ is symmetric). Let $x_i = b_i - a_i$, with $1 \leq
x_i \leq \ln n$, $k = (\ln n)^{4i - 2}$, and let $X$ and $Y$ be the $k$-th
power of $\bm{p}$ and $\bm{q}$, respectively. Let also $\nu = \ln n$
and $s = n / m$, with $(\ln n)^4 \leq s \leq n$; for simplicity, let $\nu, s \in \nats$.

We prove that for each power $k = \nu^{4i-2}$ and for both
$\bm{p}$ and $\bm{q}$, (i) the Bernoulli trials in each group $j < i$
contribute essentially $0$ to the mean value and the variance of both $X$ and
$Y$ (because for any $j < i$, $(1-a_j/\nu^{4(i+1)})^{\nu^{4i-2}} \leq n^{-\ln
n}$); (ii) the Bernoulli trials in each group $j > i$ increase the variance of
$X$ and $Y$ by at most $2 s /\nu^5$; and (iii) the difference in the mean
values of $X$ and $Y$ due to the Bernoulli trials in each group $j > i$ is
roughly $s /\nu^5$. As for the Bernoulli trials in group $i$, they contribute
roughly $s / \nu$ to the variance of $X$ and $Y$ and increase the difference in
their means by roughly $s x_i / \nu^2$. If $x_i \geq 1$, since $s \geq
(\ln n)^4 = \nu^4$ and since $n$ is sufficiently large, the difference in the
mean values of $X$ and $Y$, which is $\Omega(s / \nu^2)$, is greater than the
sum of their standard deviations, which is $O(\sqrt{s / \nu})$. Thus, by
Lemma~\ref{l:tvd_lb}, the $k$-th powers $X$ and $Y$ of $\bm{p}$ and $\bm{q}$
are at distance larger than $\eps$, a contradiction. So, an
$\eps$-approximation to the $\nu^{4i-2}$-th power of $\bm{p}$ by $\bm{q}$ is
possible only if $p_i = q_i$; thus, only if $a_i = b_i$. The details can be
found in Section~\ref{s:app:lb_lemma}. \qed
\end{proof}

To prove Theorem~\ref{th:lower_bound}, we show that given
$\eps$-approximations to the powers of $\bm{p}$ with indices $(\ln n)^{4i-2}$,
for all $i \in [m]$, we can determine the exact values of $a_1, \ldots, a_m$,
defining $\bm{p}$. Namely, given distributions $Y_1, \ldots, Y_m$, each
$Y_i$ at distance at most $\eps \leq 1/4$ to the $(\ln n)^{4i-2}$-th
power of $\bm{p}$, we can obtain a $(\ln n, (\ln n)^4, m)$-separated
vector $\bm{q}$ defined by $m$ positive integers $(b_1, \ldots, b_m)$
so that for all $i \in [m]$, the $(\ln n)^{4i-2}$-th power of $\bm{q}$
is within TVD at most $\eps$ to $Y_i$. To find such a vector $\bm{q}$,
we perform exhaustive search, in time $O((\ln n)^m)$.
That is, we try all possible tuples $(b_1, \ldots, b_m)$ and find a tuple
whose $(\ln n)^{4i-2}$-th power is within TVD at most $\eps$ to the
corresponding power $Y_i$, for all $i \in [m]$. At least one such tuple
exists, since $(a_1, \ldots, a_m)$ has this property. By the triangle
inequality, we have that for all $i \in [m]$, the $(\ln n)^{4i-2}$-th
power of $\bm{q}$ and the $(\ln n)^{4i-2}$-th power of $\bm{p}$
are at distance at most $2\eps \leq 1/2$. Thus, if the learning
algorithm succeeds in computing an $\eps$-approximation $Y_i$
to each $(\ln n)^{4i-2}$ power of $\bm{p}$, which happens with
probability at least $1-\delta \geq 1/2$, we can find an $(\ln n,
(\ln n)^4, m)$-separated vector $\bm{q}$ whose $(\ln n)^{4i-2}$-th
powers are within distance $2\eps \leq 1/2$ to the corresponding
powers of $\bm{p}$. So, by Lemma~\ref{l:lb}, we have $a_i = b_i$
for all $i \in [m]$.

Since we need $m\log\log n$ bits to represent $a_1, \ldots, a_m$ and a
sample from a PBD power has $\log n$ bits, such an $\eps$-approximation
to the powers of $\bm{p}$ with indices $(\ln n)^{4i-2}$, for $i \in [m]$,
requires $\Omega(m \log\log n / \log n)$ samples in total. Otherwise, we could
use samples from the PBD powers, provided to the learning algorithm as
input, as an economic representation of $a_1, \ldots, a_m$. Then, the learning
algorithm together with the exhaustive search procedure for finding $\bm{q}$
can be used as a ``decoding'' algorithm to obtain $a_1, \ldots, a_m$ from
their economic representation with the input samples. Since we can use any
values for $a_1, \ldots, a_m$, such a compression scheme is impossible,
see, e.g., \cite{KolmCompl}. Note, such a learning algorithm would have a certain
probability of failure, if the input samples were truly random. But here, since we know $\bm{p}$
and want to use the learning algorithm as a compression scheme
for $a_1, \ldots, a_m$, we can compute input samples deterministically, so that
the learning algorithm satisfies its error guarantees with certainty (such a sample
collection exist, since the learning algorithm has a positive probability of success).
We have thus shown that the worst-case sample complexity of any learning
algorithm for this class of instances is $\Omega(m \log\log n / \log n)$.

\section{Upper Bound for Learning PBD Powers: Separated Case}\label{s:separated_upper}

\subsection{Preliminaries}
To estimate the mean of a PBD we use the following Proposition
from \cite{DDS11}{Lemma 6}.
\begin{proposition}[Lemma 6 from \cite{DDS11}]\label{pr:mu_sampling}
  For all $n,\ \eps, \delta >0$, there exists an algorithm $\mcal{A}(n,\eps, \delta)$
  with the following properties: given access to a PBD $X$ of order $n$, it
  produces estimates $\hat{\mu}$, $\hat{\sigma}^2$ for $\mu = \Exp[X]$, $\sigma^2 =
  \Var[X]$ respectively such that with probability at least $1-\delta$:
  $|\mu - \hat{\mu} | \leq \epsilon \sigma$ and $|\sigma^2 - \hat{\sigma}^2|\leq
  \epsilon \sigma^2\sqrt{4 + \frac1{\sigma^2}}$. Moreover, $\mcal{A}$ uses
  $O(\log(1/\delta)/\eps^2)$ samples and runs in time
  $O(\log n \log(1/\delta)/\eps^2)$.
\end{proposition}

\begin{fact}\label{fact:exp_ineq_1}
  Let $m$ and $y$ be any real numbers such that $m \geq 1$ and $|y| \leq m$. Then
  we have that  $e^y (1-y^2/m)  \leq  (1+y/m)^m \leq e^y$.
\end{fact}
\begin{proof}
  See e.g. page 435 of \cite{Motwani-Raghavan}.
\end{proof}

We will give here the main ideas leading to Theorem \ref{th:separated_upper}.
The sketch of its proof can be found in Section~\ref{s:separated_upper_sketch}
and its full proof in Appendix~\ref{app:proof_theorem_separated_upper}.

Let $X \in {\cal P}$ be an unknown PBD from the defined class ${\cal P}$ (see
Section \ref{intro}). To learn $X$ means that we essentially need to learn
(approximate) values $\alpha_i$ for each $i=0,1,\ldots,s-1$. We will sample
from the following $s$ powers of $X$: $\ell_i = (c \cdot
\ln(n))^{s-i}/c$ for $i=0,1,\ldots,s-1$. We will in fact prove that all
values $\alpha_i$'s can be learned exactly.

The main idea is to learn $p_i$'s starting from the largest and proceeding
towards the smallest. Suppose that we have already exactly learned the values
of $\alpha_0, \ldots, \alpha_{i-1}$ for some $i \geq 1$. To learn $\alpha_i$ we
sample from the power $X^{\ell_i}$ by using sampler $\mcal{A}(n, \eps, \delta/s)$
from Lemma 6 in \cite{DDS11} that employs \lq\lq{}weak\rq\rq{} mean and variance
estimators (cf.~Proposition \ref{pr:mu_sampling}).
It draws $O(\log(s/\delta)/\eps^2)$ independent samples from
$X^{\ell_i}$. Since $\mu_{\ell_i}/n_i = \sum_{i=0}^{s-1} p_i^{\ell_i}$,
the error of this estimate is roughly $\varepsilon/\sqrt{n}$. Now,
$p_0,\ldots,p_{i-1}$ are precisely known, the value of $p_i^{\ell_i}$, to be
learned, is much larger than the error $\varepsilon/\sqrt{n}$, and the
values of all remaining $p_{i+1}^{\ell_i}, \ldots, p_{s-1}^{\ell_i}$ are very
small and they together sum up to at most $2/n$ and thus much smaller than the
error $\varepsilon/\sqrt{n}$. These facts let us learn $\alpha_i$
exactly. Then we proceed to the next $p_{i+1}$ and so on. The description of
the algorithm can be found as Algorithm \ref{alg:PBD_class} and the
proof of Theorem \ref{th:separated_upper} follows essentially the algorithm.
Interestingly, our algorithm solves an instance of the polynomial root
finding problem where roots $p_i$ are separated enough without using coefficients
of the polynomial, only by using the approximate powers sums
(cf.~motivation in Section \ref{intro}).

\begin{algorithm}
   \caption{Exact Learning Algorithm for Special Class of PBDs}
   \label{alg:PBD_class}
      \textbf{Input}:Random samples from powers of an unknown PBD $X \in {\cal P}$, parameter $n$,
      any constant \\ $c\geq 2$, error bound $\eps \in (0,1/(6c)]$,
      confidence bound $\delta > 0$. \\
      \textbf{Output}:Exact values of $(p_0, p_1, \ldots, p_{s-1})$
      from $X$ output with success prob.~at least $1 - \delta$.
   \begin{algorithmic}[1]


      \For{$i = 0,1, \ldots, s-1$}
      \State Call $\mcal{A}(n, \eps, \delta/s)$ (see Proposition~\ref{pr:mu_sampling})
      and draw $\bigOh{\log(s/\delta)/\eps^2}$ samples
      from $X^{\ell_i}$ to obtain $\hat{\mu}_{\ell_i}$.
      \State $\ell_i \leftarrow \frac{(c \ln(n))^{s-i}}{c}$, $n_i \leftarrow \frac{n}{s}$, $\hat{\tau}_i \leftarrow \hat{\mu}_{\ell_i}/n_i - \sum_{j=0}^{i-1} p_j^{\ell_i}$ \hspace*{1mm} (* Note: $\sum_{j=0}^{-1} p_j^{\ell_i} = 0$ *)
      \State Let $\beta_i \in \{1,2, \ldots, \lfloor\sqrt{\ln(n)}\rfloor\}$ be the smallest number s.t.
      $ (1 - \frac{\beta_i}{c \ell_i})^{\ell_i} \leq \hat{\tau}_i <  (1 - \frac{\beta_i - 1}{c \ell_i})^{\ell_i}$
      \State $a_i \leftarrow (1 - \frac{\beta_i}{c \ell_i})^{\ell_i}$, $b_i \leftarrow (1 - \frac{\beta_i - 1}{c \ell_i})^{\ell_i}$
      \State {\bf if} $\hat{\tau}_i < \frac{a_i+b_i}{2}$ {\bf then} $\alpha_i \leftarrow \beta_i$
      {\bf else} $\alpha_i \leftarrow \beta_i - 1$
      \State $p_i \leftarrow 1 - \frac{\alpha_i}{(c\cdot \ln(n))^{s-i}}$
   \EndFor
\end{algorithmic}
\end{algorithm}

\subsection{The Proof of Theorem~\ref{th:separated_upper}}
\label{s:separated_upper_sketch}
The following technical lemma will be crucial in our proof
of Theorem \ref{th:separated_upper}. Its proof can be found in
Appendix~\ref{app:proof_lemma_large_diff}.

\begin{lemma}\label{l:large_diff}
  Let $i \in \{0,1,\ldots,s-1\}$. If $\alpha_i, \beta_i \in \{1,2,\ldots,
\lfloor \sqrt{\ln(n)} \rfloor \}$ and $\alpha_i < \beta_i$, and $c \geq 2$ and
$\varepsilon \leq 1/(6c)$, then we have that $(1-\frac{\alpha_i}{c
\ell_i})^{\ell_i} - (1-\frac{\beta_i}{c \ell_i})^{\ell_i}  >
\frac{4 \varepsilon}{\sqrt{n/s}}$ for all $n \geq e^{2c}$.
\end{lemma}

We next present only a sketch of the proof, its full version
can be found in Appendix~\ref{app:proof_theorem_separated_upper}.
The proof is by induction on $i$ and we only sketch here the induction step.
Assume that for some $i \in \{1,2\ldots,s-1\}$, values of all
$p_0, p_1, \ldots, p_{i-1}$ are known exactly, and we will show how to
exactly learn $p_i$.

By Fact \ref{fact:exp_ineq_1}, we observe
$p_i^{\ell_i} \leq \left(\frac{1}{e}\right)^{1/c}$, and similarly,
$p_{i+j}^{\ell_i} \leq (1/n)^{(c \ln(n))^{j-1}}$ for
$j = 1,2, \ldots$. By geometric series properties, this implies
$\mu_{\ell_i} = \Exp[X^{\ell_i}] = n_0 \cdot \left(\sum_{j=0}^{i}
  p_j^{\ell_i} + \sum_{j=i+1}^{s-1} p_j^{\ell_i}\right) \leq
n_0 \cdot \left(\sum_{j=0}^{i-1}
  p_j^{\ell_i} + p_i^{\ell_i} + 2/n \right)$,
and letting
$\mu_{\ell_i} = n_0 \cdot \left( \sum_{j=0}^{i-1}
  p_j^{\ell_i} + p_i^{\ell_i} + r_{\ell_i} \right)$,
it implies $r_{\ell_i} \leq 2/n$.

By using inequality $(1 - x/m)^m \geq 1 - x$, true for
$m \geq 1$, $x \leq m$, and properties of geometric series
we obtain $\Var[X^{\ell_0}] =  n_0 \sum_{j=0}^{s-1}
p_j^{\ell_i} (1 - p_j^{\ell_i}) < 2 n_0$.
    Thus, by Proposition \ref{pr:mu_sampling},
$|  \mu_{\ell_i} - \hat{\mu}_{\ell_i} | \leq \varepsilon
\sigma_{\ell_i} < \varepsilon \sqrt{2 n_0}$, with probability at
least $1-\delta/s$, so
$ \left| \sum_{j=0}^{i-1} p_j^{\ell_i} + p_i^{\ell_i} +
  r_{\ell_i} - \hat{\mu}_{\ell_i}/n_0 \right| < \varepsilon \sqrt{2/n_0}$.
If we let $\alpha_i \in \{1,2, \ldots, \lfloor\sqrt{\ln(n)}\rfloor\}$ be
such that $p_i^{\ell_i} = \left(1 - \frac{\alpha_i}{c \ell_i}\right)^{\ell_i}$,
and  denote $\hat{\tau}_i = \hat{\mu}_{\ell_i}/n_0
- \sum_{j=0}^{i-1} p_j^{\ell_i}$ (recall that $p_0,
\ldots, p_{i-1}$ are known), the last inequality rewrites to
\begin{equation}\label{equ-11-main}
  \left| \left(1 - \frac{\alpha_i}{c \ell_i}\right)^{\ell_i} + r_{\ell_i} - \hat{\tau}_i \right| < \varepsilon \sqrt{2/n_0}.
\end{equation}
We can argue that there exists the smallest
$\beta_i \in \{1,2, \ldots, \lfloor\sqrt{\ln(n)}\rfloor\}$ such that
$\left(1 - \frac{\beta_i}{c \ell_i}\right)^{\ell_i} \leq \hat{\tau}_i$, so
\begin{equation}\label{equ-10-main}
   \left(1 - \frac{\beta_i}{c \ell_i}\right)^{\ell_i} \leq  \hat{\tau}_i <
   \left(1 - \frac{\beta_i - 1}{c \ell_i}\right)^{\ell_i} .
\end{equation}

Suppose now that $\beta_i \geq \alpha_i + 2$; then by (\ref{equ-10-main}), Lemma \ref{l:large_diff} and $n \geq e^{2c}$:
$\left(1 - \frac{\alpha_i}{c \ell_i}\right)^{\ell_i} + r_{\ell_i} - \hat{\tau}_i > \left(1 - \frac{\alpha_i}{c \ell_i}\right)^{\ell_i} + r_{\ell_i} -
\left(1 - \frac{\beta_i - 1}{c \ell_i}\right)^{\ell_i} \geq
$
$
\left(1 - \frac{\alpha_i}{c \ell_i}\right)^{\ell_i} + r_{\ell_i} - \left(1 - \frac{\alpha_i + 1}{c \ell_i}\right)^{\ell_i} >
\frac{4 \varepsilon}{\sqrt{n/s}}.
$ This is in contradiction with (\ref{equ-11-main}); thus, $\beta_i \leq \alpha_i + 1$.
   If  $\beta_i \leq \alpha_i - 1$, then by (\ref{equ-10-main}) and Lemma \ref{l:large_diff},
$
\left| \hat{\tau}_i - p_i^{\ell_i} - r_{\ell_i} \right|  = \left| \hat{\tau}_i - \left(1 - \frac{\alpha_i}{c \ell_i}\right)^{\ell_i} - r_{\ell_i} \right| \geq
$
$
\left| \left(1 - \frac{\alpha_i - 1}{c \ell_i}\right)^{\ell_i}  - \left(1 - \frac{\alpha_i}{c \ell_i}\right)^{\ell_i} - r_{\ell_i} \right| > \frac{4 \varepsilon}{\sqrt{n/s}} - 2/n
$ -- contradiction with (\ref{equ-11-main}). Thus, $\beta_i \in \{\alpha_i, \alpha_i + 1\}$.

By Lemma \ref{l:large_diff} the length of the interval $I_i = \left[\left(1 - \frac{\beta_i}{c \ell_i}\right)^{\ell_i}, \left(1 - \frac{\beta_i - 1}{c \ell_i}\right)^{\ell_i}\right)$ in (\ref{equ-10-main}) can be lower bounded as
$
\left(1 - \frac{\beta_i - 1}{c \ell_i}\right)^{\ell_i}  - \left(1 - \frac{\beta_i}{c \ell_i}\right)^{\ell_i} > \frac{4 \varepsilon}{\sqrt{n/s}}.
$ If $\alpha_i = \beta_i$, then by (\ref{equ-11-main}) the distance between $\left(1 - \frac{\beta_i}{c \ell_i}\right)^{\ell_i}$ and $\hat{\tau}_i$ is $\leq r_{\ell_i} + \varepsilon \sqrt{2/n_0} \leq 2/n + \varepsilon \sqrt{2/n_0}$, i.e., strictly less than half of the length of $I_i$ by $n > \frac{4}{(2-\sqrt{2})^2\varepsilon^2}$. On the other hand, if $\alpha_i = \beta_i - 1$, then by (\ref{equ-11-main}) the distance between $\left(1 - \frac{\beta_i - 1}{c \ell_i}\right)^{\ell_i}$ and $\hat{\tau}_i$ is $\leq \varepsilon \sqrt{2/n_0}$, i.e., strictly less than half of the length of $I_i$. We can use this test to decide if $\alpha_i = \beta_i - 1$ or $\alpha_i = \beta_i$. Thus the precise value of $p_i$ can be learned from $\hat{\tau}_i$.

To finish the proof, observe that by the union bound all the sampling estimates
hold with probability at least $1-\delta$. Moreover, because this sampling
for each $i=0,1,\ldots,s-1$ takes $\lceil \frac{s}{(\varepsilon^2 \delta)} \rceil$
samples from $X^{\ell_i}$, the total number of samples is
$s \cdot \lceil \frac{s}{(\varepsilon^2 \delta)} \rceil$.

\section{Upper Bound for Learning Binomial Powers}\label{s:binomial}
\subsection{Preliminaries}

To bound the Total Variation Distance of a Binomial and a PBD
we shall use the following result of Roos \cite[Theorem~2]{roos00}.
\begin{lemma}[Theorem 2 from \cite{roos00}]\label{l:roos-pbd-bin}
  Let $X = \sum_{i=1}^d$ be a PBD with probability vector $\bm{p} = (p_i)_{i=1}^n$
  and let $p \in (0,1)$. Then
  \begin{equation*}
    \dtv{X}{B(n,p)} \leq \frac{\sqrt{\me}}{2}
    \frac{\sqrt{\tau(p)}}{\lp(1-\sqrt{\tau(p)}\rp)^2 }
  \text{, where }
    \tau(p) =
    \frac{\gamma_1(p)^2+ 2 \gamma_2(p)}{2 n p (1-p)},\quad
    \gamma_j(p) =
    \sum_{i=1}^n (p-p_i)^j
  \end{equation*}
\end{lemma}
In the special case of bounding the total variation distance
of two Binomial distributions we have the following Corollary of
Lemma~\ref{l:roos-pbd-bin}
\begin{corollary}\label{c:binomial_tvd}
  Let $\eps < 1/2,\ n \geq 1$. Let $B(n,p),\ B(n,q)$ be two Binomial
  distributions such that $|p-q| \leq \eps \sqrt{\frac{p(1-p)}{n}}$ then
    $\dtv{B(n,q)}{B(n,p)} \leq 2 \sqrt{\me} \eps$.
\end{corollary}
\begin{proof}
  Following the notation of Lemma~\ref{l:roos-pbd-bin}
  we have $\gamma_1(p) \leq \eps \sqrt{n} p (1-p)$,
  $\gamma_2(p) \leq \eps^2 p(1-p)$,
  $\tau(p) \leq \eps^2/2 + \eps^2/(2n) \leq \eps^2$.
  Thus $\dtv{B(n,q)}{B(n,p)} \leq \frac{\sqrt{\me}\eps}{2 (1-\eps)^2} \leq
  2 \sqrt{\me} \eps$ when $\eps<1/2$.
\end{proof}

The proofs of the following two facts can be found in Section \ref{Chernoff_facts}.

\begin{fact}\label{fact:sampling}
  For any $\eps, \delta \in (0, 1/2)$, and $\psi > 0$, let
  $m = \lceil 4\ln(1/\delta)/(\eps^2 \psi^2) \rceil$
  and let $\p = (s_1+\cdots +s_m)/(m n)$, where $s_1, \ldots, s_m$ are $m$
  independent samples from a Binomial distribution $\B(n, p)$. Then,
  $\Prob[\p < p + \psi \err{n, p, \eps}] \geq 1-\delta,\
  \Prob[\p > p - \psi \err{n, p, \eps}] \geq 1-\delta$.
\end{fact}

\begin{fact}\label{fact:binomial_UCL}
  Let $p \in [0,1]$, $\eps, \delta \in (0, 1/2)$, $\psi > 0$,
  $k = \lceil \ln(4/\delta)/ \ln(2) \rceil$,
  $m = \lceil 4\ln(\lceil 2k/\delta \rceil)/(\eps^2 \psi^2) \rceil$.
  For $i \in [k]$ let $w_i = \sum_{i=1}^m s_i/(nm)$, with
  $s_1,\ldots,s_m$ $m$ i.i.d.~samples from $B(n,p)$.
  If $\q_1 = \min_{1 \leq i \leq k} w_i$,
  $\q_2 = \max_{1 \leq i \leq k} w_i$, then
  $\vProb{p - \psi \err{n,p,\eps} < \q_1 < p} \geq 1-\delta$,
  $\vProb{p < \q_2 < p + \psi \err{n,p,\eps}} \geq 1-\delta$.
    The overall number of samples to obtain $\q_1, \q_2$
  is $k m = O\lp(\ln(1/\delta)^2/(\eps^2 \psi^2)\rp)$.
\end{fact}

\subsection{Discussion}
We prove here that $\bigOh{1/\epsilon^2}$ samples are sufficient
to learn all the powers of a Binomial distribution $B(n,p)$ with constant
probability of success.
From Corollary \ref{c:binomial_tvd} follows that to properly learn a
Binomial distribution $B(n,p)$ within total variation distance $\bigOh{\eps}$
it's sufficient to approximate its parameter $p$ with error
$\err{n,p,\eps} = \eps \sqrt{p(1-p)/n}$.
Suppose first that the unknown $p \approx 1-1/n$,
then it is not clear at all that sampling from a constant number of powers
would suffice to approximate all the powers. We could first sample from
$B(n,p)$ to obtain an approximation $\hat{p}_1 \approx p$, but in this
case it is useless. On the other extreme, if $p \approx const$, then roughly
only the first $\log(n)$ powers matter. In fact, it is not too difficult to show
that there always exists a constant power, say $j$, such that $\hat{p}_1$
raised to power $j' \in \{j+1, j+2,\ldots, \log(n)\}$ approximates $p^{j'}$
well enough. Then we can sample from each power $i = 2, 3, \ldots, j$
separately. But how to solve the large case ($p \approx 1-1/n$) and
bridge it with the small case ($p \approx const$)?

If $p$ is large, a natural idea is to use the approximation $\hat{p}_1 \approx p$
to find a power, $\ell^*$, such that ${\hat{p}_1}^{\ell^*} \approx const$.
If we sample from $B(n,p^{\ell^*})$ and obtain an approximation
$\hat{q}_1 \approx p^{\ell^*}$, then one can argue that
${\hat{p}_j := \hat{q}_1}^{j/\ell^*}$ approximates $p^j$ well enough,
for $j=2,3,\ldots,\ell^* - 1$; that\rq{}s like using approximation $\hat{q}_1$
{\em backwards}. Similarly to the case $p \approx const$, it is possible to
show that there exists a constant power $k$ such that ${\hat{q}_1}^j$
approximates $p^{j \ell^*}$ well enough for $j \geq k+1$. The remaining
powers $j \ell^* + i$, for $j=2,\ldots,k$ and $i = 1,\ldots, \ell^* - 1$, can
be approximated by sampling from $B(n, p^{j \ell^*})$ for $j=2,\ldots,k$
(obtaining $\hat{q}_j \approx p^{j \ell^*}$), and approximating $p^{j \ell^* + i}$
by $\hat{q}_j \hat{p}_i$, where $\hat{p}_i \approx p^{i}$ was found previously,
for $i=1,2,\ldots, \ell^*$. That\rq{}s like using the approximations $\hat{q}_j$
{\em forwards}, and filling the \lq\lq{}gaps\rq\rq{} between powers
$j \ell^*$ and $(j+1) \ell^*$ by $\hat{p}_i$\rq s. It is possible to analyze
the error of such method but it features behaviour of five dimensional
functions depending on $n$, $p$, $\eps$, $\delta$ and powers $\ell$,
and thus is complex.

We show how to completely avoid these complications by generalising
our problem to allow for continuous powers, i.e., we learn $B(n,p^{\ell})$
for all $\ell \in \reals_{++}$. Considering the powers to be in $\reals_{++}$
rather than $\N$ unveils the symmetric nature of the problem. To see that,
notice that $B(n,p^{\ell})$ eventually converges to \lq\lq{}deterministic\rq\rq{}
distributions since $\lim_{\ell \to \infty} B(n,p^{\ell}) = B(n,0)$ and
$\lim_{\ell \to 0} B(n,p^{\ell}) = B(n,1)$. We are able to treat uniformly the
backwards and forwards cases which now correspond to powers smaller and greater
than one, and there is no need to fill the \lq\lq{}gaps\rq\rq{} now.
This leads to an elegant algorithm that interestingly will need to sample from
only two different powers. By Corollary~\ref{c:binomial_tvd},
the problem of approximating the
powers $B(n,p^{\ell})$ reduces to approximating $p^{\ell}$ for all $\ell \in(0,+\infty)$.
We will explain the main idea of our algorithm. Suppose that $p = 1 - 1/n + c$,
where $c < 1/n$. We split the decimal representation of $p$ in two parts. The first
part consists of roughly $\log n$ $9$'s determining the $p$\rq s closeness to $1$
(or to $0$ in the symmetric case $p \approx 0$) and the second part, referred
to as a \lq\lq{}constant\rq\rq{} part, corresponds to $c$. The decimal representation of such
a $p$ could for example be:
$p = 0.\underbrace{99\ldots9}_{\text{$\#\log n$}}\ \underbrace{458382}_{\text{\lq\lq{}constant\rq\rq{} part}}$.
It is clear that, for the first powers, the bits of $p$'s \lq\lq{}constant\rq\rq{} part are
insignificant but for higher powers $\ell$ these bits should be
learned in order to have an $\eps$-approximation of $B(n,p^{\ell})$
in total variation distance. Using Fact \ref{fact:sampling} to approximate
$p = 1 - 1/n + c$ using samples from the first power we see that we can
obtain an estimate $\hat{p}$ with precision roughly $\sqrt{p(1-p)/n} \approx 1/n$.
Thus, we learn the first $\log n$ $9$'s of the representation of $p$.
To learn the $p$'s \lq\lq{}constant\rq\rq{} part we have to sample from
a higher power to be able to distinguish the \lq\lq{}higher\rq\rq{} bits, given
the error of Fact \ref{fact:sampling}. To learn the \lq\lq{}constant\rq\rq{}
part of $p$ in our example one should sample roughly from the $n$-th
power. This idea suggests that to approximate $p$ sufficiently for \emph{all}
powers $\ell \in (0, +\infty)$, we have to obtain a good approximation of power
$-1/\log(p)$, corresponding to the number of initial $0$'s or $9$'s in
the decimal representation of $p$ and an approximation of the
\lq\lq{}constant\rq\rq{} part $c$. Our Algorithm~\ref{alg:binomial} follows
this intuition. The proof of our upper bound is based on Lemmas
\ref{l:binomial-all-powers} and \ref{l:binomial_start}.
Lemma \ref{l:binomial_start} shows that sampling from the first power
suffices to obtain an approximation
$\hat{a} = -1/\log(\hat{p})$ of $a=-1/\log(p)$. Lemma~\ref{l:binomial-all-powers}
unveils the precision up to which one should approximate $p$ to satisfy
the error $\err{n,p^{\ell},\eps}$ for all $\ell \in (0,+\infty)$, and it is the
{\em key} to handle the multidimensional analysis we face.
Algorithm \ref{alg:binomial} draws only $\bigOh{1/\eps^2}$
samples from two powers, so the overall sampling (and time) complexity is
$\bigOh{1/\eps^2}$. Note that $\psi(p^{\hat{a}})$ in Algorithm \ref{alg:binomial}
is a universal constant, as we prove below.
\begin{algorithm}
  \caption{Binomial Powers}
  \label{alg:binomial}
  \textbf{Input }: $O(\ln(1/\delta)^2/\eps^2)$ samples from the powers of $B(n,p)$.\\
  \textbf{Output }: $\hat{a},\ \hat{q}_1, \hat{q}_2$.
  \begin{algorithmic}[1]
    \State Draw $O(\ln(1/\delta)/\eps^2)$ samples from $B(n,p)$ to obtain
    the approximation $\hat{p}$ using Fact~\ref{fact:sampling}.
    \State Let $\hat{a} \gets -1/\ln(\hat{p})$.
    \State Draw $O\lp(\ln(1/\delta)^2/\lp(\eps^2 \psi(p^{\hat{a}})^2 \rp) \rp)$ samples
    from $B(n,p^{\hat{a}})$ to get estimations $\hat{q}_1,\ \hat{q}_2$ of $p$,
    $\q_1 \leq p \leq \q_2$, using Fact~\ref{fact:binomial_UCL}.
    \State \Return $\hat{a},\ \hat{q}_1,\ \hat{q}_2$
  \end{algorithmic}
\end{algorithm}

\begin{lemma}\label{l:binomial-all-powers}
  Let $\psi(p) = D\ \sqrt{\frac{p}{1-p}}\ \ln(1/p)$, where $D \approx 1.24263$.
  Let $p,\hat{q}_1, \hat{q}_2 \in(0,1)$ with $\hat{q}_1 < p < \hat{q}_2$. Then if
  $p - \hat{q}_{1} \leq  \psi(p) \err{n,p,\eps} ,\ \hat{q}_{2} - p \leq \psi(p)
  \err{n,p,\eps}$ it holds $p^l - {\q_1}^l \leq \err{n,p^l,\epsilon}$ for all
  $l \in (1, +\infty)$ and ${\q_2}^l - p^l \leq \err{n, p^l,\epsilon}$ for all $l \in(0,1)$.
\end{lemma}
\begin{proof}
  As a direct consequence of the Mean Value Theorem applied to
  the mapping $x \mapsto x^l$ we obtain
  $p^l - {\q_1}^l \leq l p^{l-1} (p-\q_1)$ for $l \in (0,1)$ and
  ${\q_2}^l - p^l \leq l p^{l-1} (\q_2 - p)$ for $l \in (1,+\infty)$.
  Next we find a function $u(p)$ such that for all $l>0$
  \begin{align}\label{inequality-power-approx}
    u(p) l p^{l-1} \err(n,p,\eps) &\leq \err(n,p^l,\eps) \\
    u(p) l p^{l-1} \sqrt{\frac{p (1-p)}{n}} &\leq \sqrt{\frac{p^l(1-p^l)}{n}} \nonumber \\
    u^2(p) l^2 p^{2l-2} p(1-p) &\leq p^l(1-p^l)\nonumber \\
    u^2(p) &\leq \frac{p}{1-p} \frac{p^{-l}-1}{l^2}\nonumber
  \end{align}
  Let $f(l) = \frac{p^{-l}-1}{l^2}$, $g(p) = 6 - 6 p^l + 4 l\ln p + l^2 (\ln p)^2$.
  Then
  \begin{align*}
    f'(l) = \frac{p^{-l} (-2 + 2 p^l - l \ln p)}{l^3} \qquad &\qquad \qquad
    f''(l) =\frac{p^{-l} (6 - 6 p^l + 4 l\ln p + l^2 (\ln p)^2)}{l^4}\\
    g'(p) &= \frac{2l ( 2- 3 p^l + l \ln(p))}{p}.
  \end{align*}
  Set $p^l = y$ and notice that the maximum of the concave function
  $y \mapsto 2- 3y + \ln(y)$ is $1-\ln(3) < 0$.
  Thus $g$ is a continuous, strictly decreasing function of $p$ and
  $\lim_{p \to 1} g(p) = 0$. Therefore $g(p) > 0$ for all $p \in (0,1)$.
  Resultantly, $f$ is a convex function of $l$ and attains its minimum
  at $\bar{l} = -\frac{C}{\ln p}$ (the root of $f'(l) = 0$),
  where $C = 2 + W_n(-2/\me^2)\footnote{$W_n$ denotes the Lambert W function.}
  \approx 1.59362$.
  It's minimum value is $f(\bar{l}) = \frac{\me^{C}-1}{C^2} (\ln p)^2$.
  Choosing
  $u(p) = D\ \sqrt{\frac{p}{1-p}}\ \ln(1/p),\ D=\frac{\sqrt{\me^{C}-1}}{C}$
  ensures that inequality~\ref{inequality-power-approx} holds.
\end{proof}
\begin{lemma}\label{l:binomial_start}
  Let $\eps \in (0,1/6),\ n\geq 1$, and $p\in (\tau,\mu)$ where
  $\tau = \frac12 \lp(1 - \sqrt{1 - 36 \eps^2/n}\rp) \leq \eps^2/n$,
   $\mu =  \frac12 \lp(1+ \sqrt{1-36\eps^2/n}\rp) \geq 1-\eps^2/n$.
  Moreover, let $a,\hat{a} \in \reals_{++}$ such that
  $p^a = \hat{p}^{\hat{a}} = 1/\me$. If $|p-\hat{p}| \leq \err{n,p,\eps}$
  then $\frac{1}{\me^2} \leq p^{\hat{a}} \leq \frac1{\me^{3/2}}.$
\end{lemma}
\begin{proof}
  Let $h = \err{n,p,\eps}$. The Taylor approximation of $f(x) = \ln(x)$
  for $x \in (p-h,p+h)$ is $\ln(x) = \ln(p) + R_0(x)$.
  Since $|f'(x) |= 1/x \leq 1/|p-h|$, we obtain
  \[\lp|\frac{R_0(x)}{\ln p} \rp| \leq \frac{1}{|\ln p|} \frac{h}{|p-h|}
    \leq \frac{1}{\lp|(1-p)p/h + p - 1\rp|}
    = \frac{1}{\lp| \frac{\sqrt{n}}{\eps}\sqrt{p(1-p)} + p -1 \rp|}.\]
  To upper bound the above quantity by $1/2$ we find the feasible set of
  the inequality $\frac{\sqrt{n}}{\eps} \sqrt{p(1-p)} \geq 3 $
  which assuming that $\eps < 1/6$ gives
  $ \frac12 \lp(1 - \sqrt{1 - 36 \eps^2/n}\rp) \leq p \leq
    \frac12 \lp(1+ \sqrt{1-36\eps^2/n}\rp)$.
  Therefore, for every $\hat{p} \in (p - h, p+h)$ we have
  \begin{align*}
    \frac12 &\leq \frac{\ln \hat{p}}{\ln p} \leq \frac32 \ \Leftrightarrow\
    -2\frac1{\ln p} \geq -\frac1{\ln\hat p} \geq - \frac23 \frac1{\ln p}\ \Leftrightarrow\
    2 a \geq \hat{a} \geq \frac23 a \ \Leftrightarrow \
    \frac{1}{\me^2} \leq p^{\hat{a}} \leq \frac1{\me^{3/2}}.
  \end{align*}
  \qed
\end{proof}

\subsection{The Proof of Theorem~\ref{th:binomial}}

Corollary \ref{c:binomial_tvd} implies that to approximate $B(n,p^{\ell})$
within total variation distance $\eps$ we need
an approximation $\p_{\ell}$ of $p^{\ell}$ with
$|p^{\ell} - \p_{\ell}| \leq \err{n,p^{\ell},\eps}$. We prove that Algorithm
\ref{alg:binomial} outputs approximations $\q_1,\ \q_2$ of $p$
satisfying this bound. We use Lemma~\ref{l:binomial_start} to show that
$1/\me^2 \leq p^{\hat{a}} \leq 1/ \me^{3/2}$, and thus,
$\psi(p^{\hat{a}}) \geq \psi(1/\me^2) = 0.983226$.
Using Fact~\ref{fact:binomial_UCL} we draw
$O\lp(\ln(1/\delta)^2/ (\eps^2)\rp)$ to obtain estimates $\q_1,\ \q_2$ such that
$\Prob\lp[\p - \err{n,p,\eps} < \q_1 < p\rp] \geq 1-\delta/2$,
$\Prob\lp[ p<\ \q_2\ < p + \err{n,p,\eps} \rp] \geq 1 -\delta/2$,
and thus the probability of success of obtaining both $\q_1, \q_2$ is
at least $1-\delta$. Having obtained the estimates $\q_1,\ \q_2$ the result
follows directly from Lemma~\ref{l:binomial-all-powers}. \qed

\begin{remark}\label{rem:binomial_integer_powers}
Algorithm~\ref{alg:binomial} can be easily modified to the case where the powers we
are allowed to sample from are natural numbers, $\ell \geq 1$. In this case
Notice that when $p \leq \me^{-C} \leq 0.2$, then the function $f$ of
Lemma~\ref{l:binomial-all-powers} is minimized for $\ell = 1$, since $f$ is convex
and the position of its global minimum is $\bar{\ell} \leq 1$. So it suffices to
choose $u(p) = \frac{p(p^{-1} - 1)}{1-p} = 1$ and we can learn all powers $\ell \geq  1$
using an estimation $\p$ obtained by sampling from the first power using
Fact~\ref{fact:sampling}. If $p \geq 0.2$ then we can simply run Algorithm \ref{alg:binomial}
with $\lceil\hat{a}\rceil$ instead of $\hat{a}$. Then, $ 0.2/\me^2 \leq p^{\lceil \hat{a} \rceil} \leq 1/\me^{3/2}$, thus $\psi(p) \geq \psi(0.2/\me^2)$, which means that
Algorithm \ref{alg:binomial} uses $\bigOh{\ln(1/\delta)/\eps^2}$ samples to
learn all powers $\ell \geq 1$.
\end{remark}

\begin{remark}\label{rem:binomial_UCL}
Algorithm \ref{alg:binomial} could use the approximation $\hat{p} \approx p$
  from Fact \ref{fact:sampling} instead of approximations $\q_1,\ \q_2$ of $p$, which
  imply a unified analysis by the Mean Value Theorem in Lemma \ref{l:binomial-all-powers}.
\end{remark}

\section{Lower Bounds for Learning Functions of Sequences of Distributions}
\label{s:minimax_lower_bounds}
\subsection{Preliminaries}
The Kullback-Leibler divergence of two probability measures $P, Q$ is
\[
  \dkl{P}{Q} = \int \log \lp(\frac{P}{Q} \rp) \mrm{d} P.
\]
Moreover, the Hellinger distance of $P,Q$ with respect to another probability
measure $\mu$ is
\[
  \dhel{P}{Q}^2 = \frac12
  \int
  \lp(
  \sqrt{ \frac{\mrm{d}P} {\mrm{d} \mu} }
  -
  \sqrt{ \frac{\mrm{d}Q} {\mrm{d}\mu} }
  \rp)^2 \mrm{d} \mu.
\]
We shall use the well known decoupling identities of Hellinger
distance and Kullback-Leibler divergence, for proofs see e.g.
\cite[Chapter 13]{duchi_stats311}.

\begin{fact}[Hellinger Decoupling Identity]\label{fact:hellinger_decoup}
  Let
  $P= P_1 \times \ldots \times P_m$ and
  $Q= Q_1 \times \ldots \times Q_m$ be two product distributions.
  Then
  \[
    \dhel{P}{Q}^2 = 1 - \prod_{i=1}^m \lp(1-\dhel{P_i}{Q_i}^2 \rp).
  \]
\end{fact}

\begin{fact}[KL-Divergence Decoupling Identity]\label{fact:kl_decoup}
  \[
    \dkl{P_1\times\ldots \times P_k}{Q_1\times \ldots \times Q_k} =
    \sum_{i=1}^k \dkl{P_i}{Q_i}.
  \]
\end{fact}
\subsection{The definition of Minimax Risk for Sequences of Distributions}
We give first notation for a definition of the minimax risk for learning
functions of sequences of distributions.  Bellow we use calligraphic letters for
sequences of distributions and gothic for sets of sequences of distributions.
Let $\mathfrak{P}$ be a family of sequences of distributions, indexed by the set $I$.
Since we can sample from every distribution $P_i$ of $\mcal{P}$ we have the
sample vector
$
X^m = (
X_{1,1}, \ldots, X_{1,{m_1}},\ \ldots, \
X_{k,1},\ldots, X_{k,{m_k}})
$
where the $i$-th group of $m_i$ samples is drawn from $P_i$, and define the multi-index
$m = (m_1,\ldots,m_k)$. All samples are independent, so $X^m$ follows the $|m|$-fold
product distribution
$P^m = P_1^{m_1} \times P_2^{m_2} \times \ldots \times P_k^{m_k}$.
Let $\theta\ : \mathfrak{P} \to \Theta$ be a function of sequences of
$\mathfrak{P}$ to be estimated. Let $\hat{\theta}\ : \mcal{X}^m \to \Theta$
be an estimator of $\theta$, and $\rho : \Theta \times \Theta \to \reals_+$
be a semimetric on the space $\Theta$. Let $d$ denote a metric in the space of
distributions. The natural choice for $d$ on the space of sequences of distributions
is to define $d(\mcal{P},\mcal{Q}) = \sup_{i \in I} d(P_i, Q_i)$.
For example we define the TVD of the
two sequences to be $\dtv{\mcal{P}}{\mcal{Q}} = \sup_{i \in I} \dtv{P_i}{Q_i}$.
\begin{definition}\label{def:adaptive_minimax_risk}
  In the above setting we define the minimax risk to be
  \begin{equation}\label{eq:adaptive_minimax_risk}
    \mathfrak{M}_N \lp(\theta(\mathfrak{P}), \rho\rp)
    \coloneqq
    \inf_{\hat{\theta}} \sup_{\mcal{P}\in\mathfrak{P}} \inf_{|m|=N}
    \Exp_{P^m} \lp[
    \rho\lp(\hat{\theta}(X^m),\ \theta(\mcal{P})
    \rp) \rp].
  \end{equation}
\end{definition}
There, the infimum over all multi-indices $m$ such that $|m| = N$
corresponds to the optimal selection of samples from each $P_i$.
Definition~\ref{def:adaptive_minimax_risk} captures the fact that
the estimator $\hat{\theta}$ can be adaptive in the
sense that \emph{after} the adversarial sequence of distributions
is picked, the optimal algorithm for the problem will choose
the best distributions from the sequence to draw samples from.

Let us give some intuitions about this definition referring to
Algorithm \ref{alg:binomial}.
This algorithm follows this definition, in that, before seeing the
input data, it samples from the first power, which then allows it
to decide from which further power to sample. Note, that the
extension of this algorithm to very large $p$ (see Appendix~\ref{s:app:binomial:hugep})
shows that the first stage of deciding from which further power to
sample can be non-trivial and requires binary search.
These operations of deciding from which powers to sample correspond to
the inner \enquote{inf} in the definition.

\subsection{Le Cam and Fano Extensions}
Let $\mcal{V}$ be a finite set of indices and let $\mfrak{F}_{\mcal{V}} \subseteq \mfrak{P}$
be a set of $|\mcal{V}|$ sequences indexed by $\mcal{V}$.
Let $V$ be the random variable representing a uniform at random
choice of a sequence of $\mfrak{F}_{\mcal{V}}$. Conditioned on the choice $V=v$,
the random sample $X^m$ is drawn from the $|m|$-fold product distribution $P_v^m$.
Let $\nu^m$ denote the joint distribution of $V, X^m$.
Let $\Psi : \mcal{X}^m \to \mcal{V} $ be a testing function, namely
$\Psi$ takes samples from the unknown sequence $\mcal{P}_V$ and outputs
an index $u\in \mcal{V}$ corresponding to a candidate sequence of distributions.
We remark that the following techniques are standard and similar derivations
can be found in \cite{yu_1997}, \cite{tsybakov_2008}, and the very good lecture notes
of John Duchi \cite{duchi_stats311}.
We are now ready to prove the standard reduction from
estimation to testing using our new definition of minimax risk.
\begin{proposition}\label{pr:estimation_testing}
  Let $\mfrak{F}_{\mcal{V}} \subseteq \mfrak{P} $ be a family of sequences of distributions
  indexed by $v\in\mcal{V}$ such that
  $\rho\lp(\theta(\mcal{P}_v, \mcal{P}_{u})\rp) \geq 2 \delta$ for all
  $\mcal{P}_v,\, \mcal{P}_u \in \mfrak{F}_{\mcal{V}}$, where,
  $v \neq u \in \mcal{V}$ and $\delta >0$.
  The minimax risk defined in Definition~\ref{def:adaptive_minimax_risk}
  has lower bound
  \[
    \mathfrak{M}_N \lp(\theta(\mathfrak{P}), \rho\rp) \geq
    \delta \inf_{m = |N|} \inf_{\Psi} \nu^m \lp(\Psi(X^m) \neq V\rp).
  \]
\end{proposition}
\begin{proof}
  Recall the definitions and the notation from Section~\ref{s:minimax_lower_bounds}
  Fix an estimator $\hat{\theta}$.
  To simplify notation we shall use $\theta$ for $\theta(\mcal{P})$ when
  the sequence $\mcal{P}$ is clear from the context, and $\theta_v$ for $\theta(\mcal{P}_v)$.
  From Markov's inequality we have
  \begin{equation}\label{eq:markov}
    \Exp_{P_v^m}\lp[\rho(\hat{\theta},\, \theta)\rp]
    \geq \delta\, P_v^m \lp( \rho(\hat{\theta},\, \theta) \geq \delta \rp)
    = \delta\, \nu^m \lp( \rho(\hat{\theta},\, \theta) \geq \delta | V = v\rp)
  \end{equation}
  Now we proceed by defining the testing function
  $
  \Psi(X^m) \coloneqq \argmin_{v \in \mcal{V}} \{ \rho(\hat{\theta}, \theta_v)\}.
  $
  Using the fact that $\rho(\theta_v, \theta_u) \geq 2\delta$ for every
  $v \neq u \in \mcal{V}$ we have that
  $ \rho(\hat{\theta}, \theta_v) \leq \delta \Leftrightarrow \Psi(\hat{\theta}) = v $.
  Now to bound the minimax risk
  \begin{align*}
    \mathfrak{M}_N \lp(\theta(\mathfrak{P}), \rho\rp) &=
    \inf_{\hat{\theta}} \sup_{\mcal{P}\in\mathfrak{P}} \inf_{|m|=N}
    \Exp_{P^m} \lp[
    \rho\lp(\hat{\theta}(X^m),\ \theta(\mcal{P})
    \rp) \rp]
    \\
    &\geq
    \inf_{\hat{\theta}}
    \sum_{v \in \mcal{V}}
    \lp(
    \frac{1}{|\mcal{V}|} \inf_{|m|=N}
    \Exp_{P_v^m} \lp[\rho(\hat{\theta},\, \theta_v) \rp]
    \rp)
    \\
    &\geq
    \delta
    \inf_{\hat{\theta}}
    \sum_{v \in \mcal{V}}
    \lp(
    \frac{1}{|\mcal{V}|} \inf_{|m|=N}
    \nu^m \lp( \rho(\hat{\theta},\, \theta_v) \geq \delta\ \vert \ V = v \rp)
    \rp)
    \\
    &=
    \delta
    \inf_{|m| = N}
    \inf_{\hat{\theta}}
    \sum_{v \in \mcal{V}}
    \lp(
    \frac{1}{|\mcal{V}|}
    \nu^m \lp( \rho(\hat{\theta},\, \theta_v) \geq \delta\ \vert\ V = v \rp)
    \rp)
    \\
    &=
    \delta
    \inf_{|m| = N}
    \inf_{\Psi}
    \nu^m \lp( \Psi(X^m) \neq V \rp),
  \end{align*}
  where for the first inequality we use the fact that,
  the supremum of a set is larger than the average of a subset of the set,
  for the second inequality we use (\ref{eq:markov}),
  and for the second equality we use the fact that
  $\inf(A + B) = \inf(A) + \inf(B)$ for any nonempty sets $A,\ B$. The last
  equality follows from Bayes' Theorem.
  \qed
\end{proof}

Using Proposition~\ref{pr:estimation_testing} we prove an extension of Le Cam's
method for sequences of distributions.
\begin{lemma}\label{l:minimax_lecam_lower_bound}
  Let $\mcal{P}, \mcal{Q} \in \mathfrak{P}$ and $\delta> 0$ such that
  $\rho\lp(\theta(\mcal{P}), \theta(\mcal{Q})\rp) \geq 2 \delta$ then
  after $N$ observations (samples) the minimax risk has lower bound
  \[
  \mathfrak{M}_N\lp(\theta(\mathfrak{P}), \rho\rp) \geq
  \frac{\delta}{2}
  ( 1 - \sqrt{2} \sqrt{1-\lp(1- \dtv{\mcal{P}}{\mcal{Q}} \rp)^N }).
\]
\end{lemma}
\begin{proof}
Since we are doing binary hypothesis testing and we want to distinguish the
distributions $P$ and $Q$ the random variable $V$ now represents the uniform choice
over the measures $P$ and $Q$.
We define the probability measure $\mu$ to be the joint distribution of
$X^m$ and $V$.
  The probability that a testing algorithm $\Psi$ outputs a wrong result in
  the binary testing problem is $\mu(\Psi(X^m) \neq V) =
  \frac12 P^m(\Psi(X^m) \neq 1) + \frac12 Q^m(\Psi(X) \neq 2)$.
  Le Cam's inequality states that
  \begin{equation}\label{eq:inequality_lecam}
    \inf_{\Psi} \lp\{P^m(\Psi(X^m) \neq 1) + Q^m(\Psi(X^m) \neq 2) \rp\} =
    1 - \dtv{P^m}{Q^m}
  \end{equation}
  Using (\ref{eq:inequality_lecam}) and Proposition~\ref{pr:estimation_testing}
  we obtain
  \[
  \mathfrak{M}_N \lp(\theta(\mathfrak{P}), \rho\rp)
  \geq \frac \delta 2 \inf_{|m| = N}
    \lp( 1 - \dtv{P^m}{Q^m} \rp) =
    \frac \delta 2
    \lp( 1 - \sup_{|m| = N}\dtv{P^m} {Q^m}\rp)
  \]
  Notice that
  \begin{align}\label{eq:tvd_hel_sup}
    \sup_{|m| = N}\dtv{P^m}{Q^m} \leq&
    \sqrt{2} \sup_{|m| = N} \sqrt{1 - \prod_{i=1}^N \lp(1- \dhel{P_i}{Q_i}^2 \rp)}
    \nonumber\\
    \leq&
    \sqrt{2} \sqrt{1-\lp(1-\sup_{i \in I}\ \dhel{P_i}{Q_i}^2\rp)^N}
    \nonumber\\
    \leq&
    \sqrt{2} \sqrt{1-\lp(1-\sup_{i \in I}\ \dtv{P_i}{Q_i}\rp)^N}
    \nonumber \\
    =&
    \sqrt{2} \sqrt{1-\lp(1-\dtv{\mcal{P}}{\mcal{Q}}\rp)^N}
  \end{align}
  where we used the inequality
  $\dhel{P}{Q}^2 \leq \dtv{P}{Q} \leq \sqrt{2}\ \dhel{P}{Q}$
  and Fact~\ref{fact:hellinger_decoup}.
  \qed
\end{proof}
Lemma~\ref{l:minimax_lecam_lower_bound} has an intuitive explanation: to distinguish
two sequences of distributions it suffices to find an index $i \in I$ such that
$\dtv{P_i}{Q_i}$ is large. Since our Definition~\ref{def:adaptive_minimax_risk} of the
minimax risk allows the algorithm to \emph{choose} the element of the sequence to draw
samples from, clearly, the hypothetical optimal algorithm of Definition
\ref{def:adaptive_minimax_risk} will choose to sample from the index where the TVD
of the two tested sequences is largest.  Therefore, to obtain a lower bound for the
testing (and thus for the estimation) problem we need to find two sequences of distributions
such that \emph{all} their elements are close in TVD but their parameters are far.

We now state Fano's Method modified to lower bound the minimax risk of
Definition~\ref{def:adaptive_minimax_risk}.
\begin{lemma}\label{l:minimax_fano_lower_bound}
  Let $\mfrak{P}$ be a set of sequences of distributions.
  Let $\mfrak{F}_{\mcal{V}} \subseteq \mfrak{P} $ be a subset
  of $\mfrak{P}$ indexed by $v\in\mcal{V}$ such that
  $\rho\lp(\theta(\mcal{P}_v), \theta(\mcal{P}_{u})\rp) \geq 2 \delta$ for all
  $\mcal{P}_v,\, \mcal{P}_u \in \mfrak{F}_{\mcal{V}}$, where,
  $v \neq u \in \mcal{V}$ and $\delta >0$.
  The minimax risk from Definition~\ref{def:adaptive_minimax_risk} has lower
  bound
  \[
    \mathfrak{M}_N \lp(\theta(\mathfrak{P}), \rho\rp) \geq
    \delta
    \lp(
    1- \frac{1}{\ln |\mcal{V}|}
    \lp(N \sup_{v,u \in \mcal{V}} \dkl{\mcal{P}_{v}}{\mcal{P}_{u}} + \ln 2 \rp)
    \rp).
  \]
\end{lemma}
\begin{proof}
Using Proposition~\ref{pr:estimation_testing} and Fano's inequality
(see e.g. \cite{cover_thomas_2006}) we can lower bound
$
\inf_{\Psi}
\nu^m \lp( \Psi(X^m) \neq V \rp)
$,
and therefore
\[
  \mathfrak{M}_N \lp(\theta(\mathfrak{P}), \rho\rp) \geq
  \delta \inf_{|m| = N} \lp(1 - \frac{I(V;X^m) + \ln 2}{\ln |\mcal{V}|} \rp)
  = \delta \lp(1 - \frac{ \sup_{|m|=N} I(V;X^m) + \ln 2}{\ln |\mcal{V}|} \rp),
\]
where $I(V;X^m)$ is the mutual information of $V,X^m$.
To upper bound the mutual information $I(V;X^m)$ we use the standard inequality
\[
  I(V;X^m) \leq \frac{1}{|\mcal{V}|^2} \sum_{v,u \in \mcal{V}} \dkl{P^m_v}{P^m_u}
  \leq \sup_{v,u \in \mcal{V}} \dkl{P^m_v}{P^m_u},
\]
which can be found in \cite{birge_1983} or \cite{yu_1997} or
page 149 of \cite{duchi_stats311}
We have
\begin{align*}
  \sup_{|m|=N} I(V;X^m)
  &\leq
  \sup_{|m|=N} \sup_{v, u \in \mcal{V}} \dkl{P^m_v}{P^m_u}
  \\
  &= \sup_{|m|=N} \sup_{v, u \in \mcal{V}}
  \dkl
  {P_{v,1}^{m_1} \times \ldots \times P_{v,k}^{m_k}}
  {P_{u,1}^{m_1} \times \ldots \times P_{u,k}^{m_k}}
  \\
  &=
  \sup_{v, u \in \mcal{V}} \sup_{m=|N|}
  \sum_{i = 1}^k m_i \dkl{P_{v,i}}{P_{u,i}}
  \\
  &\leq N \sup_{v,u \in \mcal{V}, i \in I} \dkl{P_{v,i}}{P_{u,i}}
  \\
  &= N \sup_{v,u \in \mcal{V}} \dkl{\mcal{P}_{v}}{\mcal{P}_{u}},
\end{align*}
where to obtain the second equality we use Fact~\ref{fact:kl_decoup}
\end{proof}

\subsection{Applications}
\subsubsection{Application 1: Parameter Estimation for PBDs
  (Theorem~\ref{th:lower_bound_exp}).}
Since we are estimating the parameters of the PBD, Le Cam's method is well
suited for this problem.  To prove Theorem~\ref{th:lower_bound_exp},
using Lemma~\ref{l:minimax_lecam_lower_bound}, we extend the argument given
in Proposition 15 of \cite{DKS16colt2} to prove that $\Omega(2^{1/\eps})$ samples
are required even in the case where we are allowed to sample from the powers of
the PBDs. The key idea is that the instance used in their proof suffices to prove
that the TVD of \emph{all} the powers is $O(2^{-1/\eps})$ whereas the separation
of the parameter vector is $\Omega(\eps)$.

Using the notation introduced in the beginning of this section we
denote by $\mcal{P}$ the sequence of the powers of a PBD and since we
want to estimate the parameters $p_i$ we have $\theta(\mcal{P}) = \bm{p}$.
Our metric in the space of vectors, $(0,1)^n$, is $\rho(p, \hat{p})  = \|p-\hat{p}\|_{\infty}$
since we want to approximate the vector $\bm{p}$ in additive error at
most $\eps$.

We will follow the argument given in the proof of Proposition 15 in
\cite{DKS16colt2} and therefore we will present it fully for the sake
of completeness. We set the length $n$ of the PBD vector to be
$n = \Theta(\log(N/\eps))$ where $N$ represents the number of samples
in the minimax risk definition. We take
$p_j \coloneqq (1+\cos(\frac{2 \pi j}{n}))/8$,
$q_j \coloneqq (1+\cos(\frac{2 \pi j+\pi}{n}))/8$, $j \in [n] $.
Then for $j = n/4 + O(1)$, we have that $|p_i - p_j| = \Omega(1/\log(N/\eps))$
since for all $i$, $\frac{ 2 \pi i + \pi}{n}$ is at least $\Omega(1/\log(N/\eps))$
from $\frac{2 \pi i}{n}$ and $\frac{2 \pi (n-j)}{n}$.

Observe that $p_1,\ldots,p_n$ resp. $q_1,\ldots,q_n$ are roots of the
Chebyshev's polynomials, $(T_n(8x -1)-1)$, resp. $(T_n(8x-1)+1)$, where
$T_n$ is the $n$-th Chebyshev polynomial.
Since these polynomials agree in all coefficients except from their constant
terms, the Newton-Girard identities imply that $\sum_{i=1}^n p_i^l = \sum_{i=1}^n q_i^l$
for all $l \in \{1,2,\ldots,n-1\}$ and moreover, for $l \geq n$
it is easy to see that $3^l (\sum_{i=1}^n (p_i^l-q_i^l)) \leq n (3/4)^n =
\log(N/\eps) (3/4)^{\log(N/\eps)} $.
For small enough $\eps$ using Lemma 9 of \cite{DKS16colt2} we have that
$\dtv{P_1}{Q_1} \leq c/N$ for some constant $c$. We will show that this in fact is true
for all powers $P_s, Q_s$ of these two PBDs, To show this let us fix any power
$s \in \{1,2, \ldots, n\}$.
Then, we have that $\sum_{i=1}^n p_i^{s l} = \sum_{i=1}^n q_i^{s l}$, for
any $l = 1,2, \ldots, \lfloor (n-1)/s\rfloor$ assuming that $s \leq n-1$. Moreover, when
$l \in \{\lfloor (n-1)/s\rfloor +1, \lfloor (n-1)/s\rfloor + 2, \ldots\} $, we
have $3^l (\sum_{i=1}^n (p_i^{s l}-q_i^{s l})) \leq n \frac{3^l}{4^{s l}} \leq
n \frac{3^{s l}}{4^{s l}} \leq n (\frac{3}{4})^n$, where the last inequality
holds because $sl \geq n$. It is easy to see that the same hold when $s=n$, and
once more by Lemma 9 in \cite{DKS16colt2}, $\dtv{P_s}{Q_s} \leq c/N$.
Since the separation of the parameters is $\Omega(1/\log(N/\eps))$ and
the Total Variation distance of the two sequences is less that $c/N$
we can use Lemma~\ref{l:minimax_lecam_lower_bound} to obtain a minimax lower bound
rate of $1/\log(\eps/N)$. Notice that an upper bound of $c/N$ on the
total variation distance of the two sequences implies a lower bound of
$\Omega(\delta)$ for the minimax risk. Therefore, since we need to approximate
the parameters to additive error $\eps$, $\mfrak{M}_N < \eps$ implies that the
number of samples $N$ should be $\Omega(2^{1/\eps})$.

\subsubsection{Application 2: Learning Powers of Binomials.
  (Theorem~\ref{th:binomial_lower_bound})}
We use Lemma~\ref{l:minimax_fano_lower_bound} to prove a matching
lower bound of $\Omega(1/\eps^2)$ for the problem of learning the powers
of a Binomial distribution. Its quite technical proof
(cf.~Section~\ref{s:app:binomial_lower}), is only sketched here.

Using the notation of Lemma~\ref{l:minimax_fano_lower_bound} and since we do
density estimation, we have $\theta(\mcal{P}) = (f_i)_{i\in \nats}$.
Therefore we will use the metric $
\rho(\theta(\mcal{P}),\theta(\mcal{Q})) = \dtv{\mcal{P}}{\mcal{Q}} =
\sup_{i \in \nats} d_{\mrm{tv}} (f_i, \hat{f}_i)
$.\\
Let $\delta = \Theta(1/\sqrt{n\, N})$. Let $p_1 = 1/2$,
$p_2 = 1/2 + \delta/4$, $p_3 = 1/2 + \delta/2$.
Let $\mcal{P}_1 = (B(n,(1/2)^i))_{i\in \N}$,
$\mcal{P}_2 = B(n,(1/2+\delta/4)^i)_{i \in \N}$,
$\mcal{P}_3 = B(n,(1/2+\delta/2)^i)_{i \in \N}$.
The TVD of the first powers of these Binomials is $\Omega(1/\sqrt{N})$.
To see this notice that since the variance of the Binomials is
$O(n)$ we can approximate the Binomials with Normals with insignificant error.
When their variances are close, the TVD of two Normals is roughly
proportional to the difference of their means divided by their \enquote{common}
standard deviation, which is $\Omega(1/\sqrt{N})$.
Thus we obtain our lower bound for the TVD.  We then prove an upper bound for
the KL-divergence between \emph{all} powers, namely
$\dkl{\mcal{P}_1}{\mcal{P}_3} = O(1/N)$.  It's easy to see that this upper bound
holds for the first power.  To prove that it holds for all the powers notice
that $\dkl{B(n,p)}{B(n,q)}$ is an increasing function of $|p-q|$
(for a proof see Proposition~\ref{pr:binomial_kl_monotonicity}).
Thus, since the distances of the $p_i$'s of our three Binomials roughly decrease
for higher powers the KL-Divergence of the first power is roughly an upper
bound for $\dkl{\mcal{P}_1}{\mcal{P}_3}$.
Now, applying Lemma~\ref{l:minimax_fano_lower_bound} we have that
$
\mathfrak{M}_N \lp(\theta(\mathfrak{P}), \rho\rp) = \Omega(1/\sqrt{N}), $
which in turn implies that to have an estimator that approximates
all the powers in distance less than $\epsilon$ we need
$
\mathfrak{M}_N \lp(\theta(\mathfrak{P}), \rho\rp) < \epsilon
$
and therefore the number of samples $N$ should be $\Omega(1/\eps^2)$.
\qed

\section{Upper Bound for Parameter Estimation}
Newton's identities, a.k.a~the Newton-Girard formulae, give relations between
power sums and elementary symmetric polynomials of variables $x_1,...,x_n$. In
that setting, the $k$th power sum is $s_k(x_1, \ldots, x_n) = x_1^k +\cdots+
x_n^k$. The $k$th elementary symmetric polynomial $e_k(x_1, \ldots, x_n)$ is
the sum of all distinct products of $k$ distinct variables. Newton's identities
allow us to compute the elementary symmetric polynomials if we know the power
sums \emph{exactly}. Moreover, the polynomial with roots $x_i$, i.e.,
$\prod_{i=1}^n (x-x_i)$, may be expanded as $\sum_{k=0}^n (-1)^{n+k} e_{n-k}
x^k$. Thus, if we know the power sums $s_1(x_1, \ldots, x_n), \ldots, s_n(x_1,
\ldots, x_n)$ exactly, we can first find the coefficients of the elementary
symmetric polynomials and then compute the roots $x_1, \ldots, x_n$ with an
arbitrarily good accuracy. A similar approach was used in \cite{DP13} to
derive sparse covers for PBDs.

In this section we provide the analysis of the  \enquote{noisy} version of Newton's
identities. Given query access to PBD powers, we can obtain good
estimations of the power sums $s_k(p_1, \ldots, p_n)$ using a reasonable number
of samples, since the expectations of PBD powers are the power sums of the
unknown probabilities $p_1, \ldots, p_n$. An intriguing question is to which
extent these \enquote{noisy} power sum estimations can be used to recover the actual
values of $p_1, \ldots, p_n$ within sufficiently good accuracy. In this Section
we answer this question by providing an upper bound on the sampling complexity
of estimating the parameters of a PBD using samples from its powers.
This upper bound matches the corresponding lower bound of
Theorem~\ref{th:lower_bound_exp}.

\subsection{Preliminaries}
\ignore{
} 

Let $\mathbf{x} \in \reals^n$ be a vector, and
$\mathbf{A} = (\bm{A}_{ij})_{i,j \in [n]}$ be a $n\times n$ matrix. Then
$\|\bm{x}\|_{\infty} = \max_{i \in [n]} |\bm{x}_i|$,
$\|\bm{A}\|_{\infty} = \max_{i\in[n]}\sum_{j=1} a_{ij}$,
$| \bm{x} | = (|\bm{x}_i|)_{i\in[n]}$,
$| \bm{A} | = (|\bm{A}_{ij}|)_{i,j \in [n]}$.
We use \enquote{$\leq$} in $\bm{A} \leq \bm{B}$ to denote element-wise
inequality of the matrices $\bm{A}$, $\bm{B}$, namely
$\bm{A} \leq \bm{B} \Leftrightarrow A_{ij} \leq B_{ij}$
for all $i, j \in [n]$.

To compute the sensitivity of the solution of a linear system $\bm{A} \bm{x} = \bm{b}$
to perturbations of $\bm{A}, \bm{b}$ we shall use Theorem 7.4
from \cite{higham_accuracy_2002}, formulated bellow as Lemma~\ref{l:linear_perturbations}.
\begin{lemma}[Theorem 7.4 from \cite{higham_accuracy_2002}]\label{l:linear_perturbations}
  Let $\bm{A}\bm{x} = \bm{b}$ and
  $(\bm{A} + \Delta \bm{A}) \bm{y} = \bm{b} + \Delta \bm{b}$, where
  $|\Delta \bm{A}| \leq u\ \bm{E}$ and $|\Delta \bm{b}| \leq u\ \bm{f}$, and
  assume that $u\\ \|\, |\bm{A}^{-1}|\ \bm{E}\, \| < 1$, where $\|\cdot\|$ is
  an absolute norm. Then
  \[ \frac {\|\bm{x}-\bm{y}\|}{\|\bm{x}\|} \leq
    \frac {u} {1 - u\, \|\, |\bm{A}^{-1}|\ \bm{E}\, \|}
    \frac {\|\ |\bm{A}^{-1}|\ (\bm{E}|\bm{x}| + \bm{f})\ \|} {\|\bm{x}\|}
  \]
  and for the $\infty$-norm this bound is attainable to first order in $u$.
\end{lemma}

To approximate the roots of the univariate polynomial $P(x)$ we
use the nearly optimal root finding algorithm of Pan \cite{Pan2001}(Theorem 2.1.1) formulated
below as Lemma~\ref{l:pan_root_finding}.
\begin{lemma}[Theorem 2.1.1 from \cite{Pan2001}]\label{l:pan_root_finding}
  Let $P(x) = \sum_{i=0}^n c_i x^i = c_n \prod_{i=1}^n (x-p_i)$, $c_n \neq 0$,
  be a polynomial of degree $n$ such that all its complex roots satisfy
  $ |p_j| \leq 1$  for all $j$ .
  Let $b$ be a fixed real number, $b \geq n \log n$. Then complex
  numbers $\hat{p}_j$ can be computed
  by using $\bigOh{(n \log^2 n) ( \log^2 n + \log b)}$ arithmetic
  operations performed with the precision of $\bigOh{b}$ bits such that
  $|\hat{p}_j - p_j| < 2^{2- b/n}$ for $j = 1,\ldots, n$.
\end{lemma}

The following simple bound on the coefficient vector of a polynomial with
roots in $[-1,1]$ will be useful.
\begin{fact}\label{fact:pol_coef_bound}
  If all roots of a monic polynomial
  $P=x^n + a_{n-1}x^{n-1} +...+a_0$ of degree $n$
  lie in the interval $[-1,1]$ then $|a_k| \leq \binom{n}{k} \leq 2^n$.
\end{fact}
\begin{proof}
  Using Vieta's formulae we have
  \[
    a_{n-k} = (-1)^k \sum_{1\leq i_1 < i_2 < \ldots < i_k \leq n}
    x_{i_1} x_{i_2} \ldots x_{i_k}
  \]
  Therefore, $|a_{n-k}|$ is maximum when all $x_i$ are $1$ and therefore
  $|a_{n-k}| \leq \binom{n}{k} = \binom{n}{n-k}$.
\end{proof}

\subsection{The proof of Theorem~\ref{th:newton_identities_samples}}
\label{s:app:newton_identities_samples}
We denote by $P_j$ the $j$-th power of the PBD with probability
vector $\bm{p}$, namely $P_j$ is the PBD with probability vector
$\bm{p}^j = (p_i^j)_{i=1}^n$.
Let $P(x) = x^n + c_{n-1} x^{n-1} + \ldots + c_0 = \prod_{i=1}^n (x-p_i)$
be the monic polynomial of degree $n$. Notice that the mean value of $P_j$
denoted by $\mu_j$ equals the $j$-th Newton sum of the roots of $P(x)$ since
$\mu_j = \sum_{i=1}^n{p_i^j}$. Given that $P(x)$ is monic, the coefficients
of $P(x)$ and the means $\mu_1, \mu_2, \ldots, \mu_n$ satisfy the following
linear system of Newton's identities:
\begin{eqnarray*}
  \mu_j + \sum_{i=1}^{j-1} c_{n-i}\mu_{j-i} + j c_{n-j} =0, \quad j = 1,2, \ldots,n \\
\end{eqnarray*}
The system has the following matrix form, where we omit zero elements.
\begin{equation}\label{eq:newton_identities}
  \begin{pmatrix}
    1   &  &  & & \\
    \mu_1 & 2 &  &  &  \\
    \mu_2 & \mu_1 & 3 & \\
    \vdots & \vdots &\ddots & \ddots & \\
    \mu_{n-1} & \mu_{n-2} & \dots &  \mu_1 & n
  \end{pmatrix}
  \begin{pmatrix}
    c_{n-1} \\
    c_{n-2} \\
    c_{n-3} \\
    \vdots \\
    c_0
  \end{pmatrix}
  =
  \begin{pmatrix}
    - \mu_{1} \\
    - \mu_{2} \\
    - \mu_{3} \\
    \vdots \\
    - \mu_n
  \end{pmatrix}
  \Leftrightarrow \bm{A} \bm{c} = \bm{b}
\end{equation}

Using the linear system \ref{eq:newton_identities} Algorithm~\ref{alg:parameter}
retrieves an approximation of the coefficient vector $\bm{c}$ and then finds
the roots of the corresponding univariate polynomial.
\begin{algorithm}
  \caption{Parameter Estimation}
  \label{alg:parameter}
  \textbf{Input}: $2^{\bigOh{n \max(\log(1/\eps),\log(n)}}$ samples from the
  powers $P_j$, $j \in [n]$.\\
  \textbf{Output}: An additive $\eps$ approximation of $\bm{p}$.
  \begin{algorithmic}[1]
    \State Using $\mcal{A}$ of Lemma~\ref{pr:mu_sampling}
    draw $2^{\bigOh{n \max\lp( \log(1/\eps), \log n \rp)}}$ samples from
    each power $P_j$ to obtain the approximations $\hat{\mu}_j$ of $\mu_j$.
    \State Solve the system \ref{eq:newton_identities}
    and obtain $\hat{\bm{c}}$.
    \State Use Pan's Algorithm of Lemma~\ref{l:pan_root_finding}
    to compute approximations $\hat{\bm{p}}_j$ to all the roots of the
    polynomial $P(x) = \sum_{i=1}^n c_i x^i$.
    \State \Return $\hat{\bm{p}}$.
  \end{algorithmic}
\end{algorithm}

We now proceed to the proof of Theorem~\ref{th:newton_identities_samples}.

Starting from the last step of root finding with Pan's Algorithm of
Lemma~\ref{l:pan_root_finding}, we have that, to obtain
$\eps$-approximations of the roots of the polynomial $P(x)$ we need
to obtain an approximating vector $\hat{\bm{c}}$ of the coefficient vector $\bm{c}$
of $P(x)$ such that
\begin{equation}\label{eq:coefficient_precision}
  \|\bm{c} - \hat{\bm{c}} \|_{\infty} = 2^{\bigOh{-n \max(\log(1/\eps),\log(n)}}.
\end{equation}

Next, we proceed to computing the precision needed for the means $\mu_j$
so that the system of Newton Identities (\ref{eq:newton_identities})
can be solved to provide a solution satisfying (\ref{eq:coefficient_precision}).
Since in our setting the error of approximating the $j$-th mean is proportional
to the standard deviation of the $j$-th powers, the errors $\bm{E}$,
$\bm{f}$ of Lemma~\ref{l:linear_perturbations} are
\[
  \bm{E}  =
  \begin{pmatrix}
    \sigma_1 & & & \\
    \sigma_2 & \sigma_1 & \\
    \vdots & \vdots & \ddots & \\
    \sigma_n & \sigma_2 & \dots & \sigma_1
  \end{pmatrix}
  \leq \sqrt{n}
  \begin{pmatrix}
    1 & & & \\
    1 & 1 & \\
    \vdots & \vdots & \ddots & \\
    1 & 1 & \dots & 1
  \end{pmatrix}
  \qquad
  \bm{f} =
  \begin{pmatrix}
    \sigma_{1} \\
    \sigma_{2} \\
    \sigma_{3} \\
    \vdots \\
    \sigma_n
  \end{pmatrix}
  \leq
  \sqrt{n}
  \begin{pmatrix}
    1 \\
    1 \\
    1\\
    \vdots \\
    1
  \end{pmatrix}
\]
Since it holds $\mu_j \leq n$ and it follows that $\bm{A}_{ij} \leq n$.
Since $\bm{A}$ is lower triangular, $\det(\bm{A}) = n!$, and it holds that
$\det(\bm{A}) \geq M_{ij}$, where $M_{ij}$ is the determinant of
of $(n-1)\times(n-1)$ submatrix of $\bm{A}$ after deleting row $i$ and
column $j$.
It follows that $|\bm{A}^{-1}|_{ij} \leq 1$. Moreover, since the solution
vector $x$ corresponds to the coefficients of $P(x)$ from
Fact~\ref{fact:pol_coef_bound} it follows $|x|_i \leq \binom{n}{n-i}$.
Using these inequalities we bound
\begin{align*}
  |\bm{A}^{-1}||\bm{E}| \leq
  \sqrt{n}
  \begin{pmatrix}
    1 & & & \\
    1 & 1 & \\
    \vdots & \vdots & \ddots & \\
    1 & 1 & \dots & 1
  \end{pmatrix}
  \begin{pmatrix}
    1 & & & \\
    1 & 1 & \\
    \vdots & \vdots & \ddots & \\
    1 & 1 & \dots & 1
  \end{pmatrix}
  =
  \sqrt{n}
  \begin{pmatrix}
    1 & & & \\
    2 & 1 & \\
    \vdots & \vdots & \ddots & \\
    n & n-1 & \dots & 1
  \end{pmatrix}
\end{align*}
Moreover, $|\bm{A}^{-1}| |\bm{f}| \leq (n\ 2n\ \ldots\ n^2)^T$. Combining the above
inequalities we can estimate the condition of $\bm{A}$.
\begin{align*}
  \|\, |\bm{A}^{-1}|\, |\bm{A}|\, |\bm{c}| + |\bm{A}^{-1}|\, |\bm{b}|\, \|_\infty
  &\leq \sqrt{n} \sum_{i=0}^n (n-i)\binom{n}{n-i} + \sqrt{n} n = n^{3/2} (2^{n-1} + 1)
  = \bigOh{n^{3/2} 2^n}\\
  \| |\bm{A}^{-1}| |\bm{E}| \|_\infty &= \bigOh{n^{5/2}}
\end{align*}
Thus, from Lemma~\ref{l:linear_perturbations} we obtain the following
absolute error bound with respect to the $\infty$-norm
\begin{equation*}
  \|\bm{c} -\hat{\bm{c}}\|_\infty \leq u \ \bigOh{n^{3/2}2^n}
\end{equation*}

Since we need to run Algorithm $\mcal{A}$ of Proposition~\ref{pr:mu_sampling}
$n$ times to obtain approximations $\hat{\mu_j}$ such that
$|\mu_j - \hat{\mu_j}| \leq u_j \sigma_j$ for all $j\in[n]$ with probability
at least $1-\delta$ it follows from the union bound that we have to draw
$\bigOh{\log(1/n) n /u^2}$ from the powers $P_j$, $j \in [n]$.
Therefore, since $u \bigOh{n^{3/2} 2^n}$ should satisfy
(\ref{eq:coefficient_precision}) we conclude that overall we need
$ 2^{\bigOh{n \max(\log(1/\eps),\log(n)}}$  samples.

\bibliography{learning}

\begin{thebibliography}{10}

\bibitem{AD15}
J.~Acharya and C.~Daskalakis.
\newblock {Testing Poisson Binomial Distributions}.
\newblock In {\em Proc. of the 26th {ACM-SIAM} Symposium on Discrete Algorithms
  (SODA~'15)}, pages 1829--1840, 2015.

\bibitem{BH11}
M.{-}F. Balcan and N.J.A. Harvey.
\newblock Learning submodular functions.
\newblock In {\em Proc. of the 43rd {ACM} Symposium on Theory of Computing
  (STOC~'11)}, pages 793--802, 2011.

\bibitem{Bir97}
L.~Birg\'e.
\newblock Estimation of unimodal densities without smoothness assumptions.
\newblock {\em Annals of Statistics}, 25(3):970--981, 1997.

\bibitem{birge_1983}
Lucien Birg\'e.
\newblock Approximation dans les espaces métriques et théorie de
  l'estimation.
\newblock {\em Zeitschrift für Wahrscheinlichkeitstheorie und Verwandte
  Gebiete}, 65(2):181--237, December 1983.

\bibitem{JTChu55}
J.T. Chu.
\newblock On bounds for the normal integral.
\newblock {\em Biometrika}, 42:263--265, 1955.

\bibitem{cover_thomas_2006}
Thomas~M. Cover and Joy~A. Thomas.
\newblock {\em Elements of {Information} {Theory} 2nd {Edition}}.
\newblock Wiley-Interscience, Hoboken, N.J, 2 edition edition, July 2006.

\bibitem{DDKT16}
C.~Daskalakis, A.~De, G.~Kamath, and C.~Tzamos.
\newblock A size-free {CLT} for {Poisson} multinomials and its applications.
\newblock In {\em Proc. of the 48th {ACM} Symposium on Theory of Computing
  (STOC~'16)}, pages 1074--1086, 2016.

\bibitem{DDS13}
C.~Daskalakis, I.~Diakonikolas, R.~O'Donnell, R.A. Servedio, and L.{-}Y. Tan.
\newblock Learning sums of independent integer random variables.
\newblock In {\em Proc. of the 54th {IEEE} Symposium on Foundations of Computer
  Science (FOCS~'13)}, pages 217--226, 2013.

\bibitem{DDS11}
C.~Daskalakis, I.~Diakonikolas, and R.A. Servedio.
\newblock {Learning Poisson Binomial Distributions}.
\newblock {\em Algorithmica}, 72(1):316--357, 2015.

\bibitem{DKT15}
C.~Daskalakis, G.~Kamath, and C.~Tzamos.
\newblock {On the Structure, Covering, and Learning of Poisson Multinomial
  Distributions}.
\newblock In {\em Proc. of the 56th {IEEE} Symposium on Foundations of Computer
  Science (FOCS~'15)}, pages 1203--1217, 2015.

\bibitem{DP15}
C.~Daskalakis and C.H. Papadimitriou.
\newblock Approximate {Nash} equilibria in anonymous games.
\newblock {\em J. Economic Theory}, 156:207--245, 2015.

\bibitem{DP13}
C.~Daskalakis and C.H. Papadimitriou.
\newblock {Sparse Covers for Sums of Indicators}.
\newblock {\em {Probability Theory and Related Fields}}, 162(3):679--705, 2015.

\bibitem{DS16}
C.~Daskalakis and V.~Syrgkanis.
\newblock {Learning in Auctions: Regret is Hard, Envy is Easy}.
\newblock In {\em Proc. of the 57th {IEEE} Symposium on Foundations of Computer
  Science (FOCS~'16)}, 2016.

\bibitem{Diak16}
I.~Diakonikolas.
\newblock Learning structured distributions.
\newblock In P.~B{\"{u}}hlmann, P.~Drineas, M.~Kane, and M.J. van~der Laan,
  editors, {\em Handbook of Big Data}, pages 267--283. Chapman and Hall/CRC,
  2016.

\bibitem{DKS16colt2}
I.~Diakonikolas, D.M. Kane, and A.Stewart.
\newblock Properly learning poisson binomial distributions in almost polynomial
  time.
\newblock In {\em Proceedings of the 29th Conference on Learning Theory,
  ({COLT}'16)}, pages 850--878, 2016.

\bibitem{DKS16colt1}
I.~Diakonikolas, D.M. Kane, and A.~Stewart.
\newblock Optimal learning via the fourier transform for sums of independent
  integer random variables.
\newblock In {\em Proceedings of the 29th Conference on Learning Theory,
  ({COLT}'16)}, pages 831--849, 2016.

\bibitem{DKS16}
I.~Diakonikolas, D.M. Kane, and A.~Stewart.
\newblock {The {Fourier} Transform of {Poisson} Multinomial Distributions and
  its Algorithmic Applications}.
\newblock In {\em Proc. of the 48th {ACM} Symposium on Theory of Computing
  (STOC~'16)}, pages 1060--1073, 2016.

\bibitem{DP09}
D.P. Dubhashi and A.~Panconesi.
\newblock {\em Concentration of Measure for the Analysis of Randomized
  Algorithms}.
\newblock Cambridge University Press, 2009.

\bibitem{duchi_stats311}
John Duchi.
\newblock Stats311, {Lecture} {Notes}.
\newblock URL:
  \url{https://stanford.edu/class/stats311/Lectures/full_notes.pdf}.

\bibitem{FK14}
V.~Feldman and P.~Kothari.
\newblock Learning coverage functions and private release of marginals.
\newblock In {\em Proc. of the 27th Conference on Learning Theory (COLT~2014)},
  volume~35 of {\em {JMLR} Proceedings}, pages 679--702, 2014.

\bibitem{higham_accuracy_2002}
Nicholas~J. Higham.
\newblock {\em Accuracy and {Stability} of {Numerical} {Algorithms}}.
\newblock SIAM: Society for Industrial and Applied Mathematics, Philadelphia,
  2nd edition edition, August 2002.

\bibitem{ito_diffusion_1996}
Kiyosi Itô and Henry P.~jr Mckean.
\newblock {\em diffusion {processes} and their {sample} {paths}}.
\newblock springer, berlin ; new york, 1996 edition edition, February 1996.

\bibitem{KolmCompl}
M.~Li and P.M.B. Vit{\'{a}}nyi.
\newblock {\em An Introduction to Kolmogorov Complexity and Its Applications,
  3rd Edition}.
\newblock Texts in Computer Science. Springer, 2008.

\bibitem{CGS2010}
Qi-Man~Shao Louis H.Y.~Chen, Larry~Goldstein.
\newblock {\em Normal Approximation by Stein's Method}.
\newblock Springer, 2010.

\bibitem{Motwani-Raghavan}
R.~Motwani and P.~Raghavan.
\newblock {\em Randomized Algorithms}.
\newblock Cambridge University Press, 1995.

\bibitem{Pan2001}
Victor Pan.
\newblock Univariate polynomials: Nearly optimal algorithms for numerical
  factorization and rootfinding.
\newblock {\em Journal of Symbolic Computation}, 33(5):701-733, 2002.

\bibitem{roos00}
B.~Roos.
\newblock Binomial {Approximation} to the {Poisson Binomial Distribution}: The
  {Krawtchouk} expansion.
\newblock {\em Theory of Probability and its Applications}, 45(2):328--344,
  2000.

\bibitem{tsybakov_2008}
Alexandre~B. Tsybakov.
\newblock {\em Introduction to {Nonparametric} {Estimation}}.
\newblock Springer, New York ; London, 1 edition edition, November 2008.

\bibitem{wald_1939}
Abraham Wald.
\newblock Contributions to the {Theory} of {Statistical} {Estimation} and
  {Testing} {Hypotheses}.
\newblock {\em The Annals of Mathematical Statistics}, 10(4):299--326, December
  1939.

\bibitem{yang_barron_1999}
Yuhong Yang and Andrew Barron.
\newblock Information-theoretic determination of minimax rates of convergence.
\newblock {\em The Annals of Statistics}, 27(5):1564--1599, October 1999.

\bibitem{yang_2015}
Zhen-Hang Yang and Yu-Ming Chu.
\newblock On approximating {Mills} ratio.
\newblock {\em Journal of Inequalities and Applications}, 2015:273, September
  2015.

\bibitem{yu_1997}
Bin Yu.
\newblock Assouad, {Fano}, and {Le} {Cam}.
\newblock In {\em Festschrift for {Lucien} {Le} {Cam}}, pages 423--435.
  Springer, New York, NY, 1997.

\end{thebibliography}

\appendix\makeatletter
\edef\thetheorem{\expandafter\noexpand\thesection\@thmcountersep\@thmcounter{theorem}}
\makeatother

\section{Appendix -- Lower Bound for Learning PBD Powers: Separated Case}
\subsection{The Proof of Lemma~\ref{l:tvd_lb}}\label{s:app:tvd_lb}
For simplicity, we let $\lambda \equiv 2\sqrt{\ln(\frac{2}{1-\eps})}$ and
assume that $\mu_Y > \mu_X + \lambda(\sigma_X + \sigma_Y)$ (the other case is
symmetric). By the definition of TVD, we obtain that:
\begin{align*}
  2 \dtv{X}{Y} = & \sum_{i=0}^\infty | \Prob[X = i] - \Prob[Y = i] | \\
  \geq & \sum_{i=0}^{\mu_X + \lambda \sigma_X} ( \Prob[X = i] - \Prob[Y = i]) +
  \sum_{i=\mu_Y - \lambda \sigma_Y}^{\infty} ( \Prob[Y = i] - \Prob[X = i])\\
  = &\ (\Prob[X \leq \mu_X + \lambda \sigma_X] -
  \Prob[Y \leq \mu_X + \lambda \sigma_X])\,+\\
  &\,+ (\Prob[Y \geq \mu_Y - \lambda \sigma_Y] -
  \Prob[X \geq \mu_Y - \lambda \sigma_Y]) \\
  \geq &\ (1 - \Prob[X > \mu_X + \lambda \sigma_X] -
  \Prob[Y < \mu_Y - \lambda \sigma_Y])\,+ \\
  &\,+ (1 - \Prob[Y < \mu_Y - \lambda \sigma_Y] -
  \Prob[X > \mu_X + \lambda \sigma_X]) \\
  > &\ (1 - (1-\eps)/2 - (1-\eps)/2) +  (1 - (1-\eps)/2 - (1-\eps)/2) = 2\eps
\end{align*}
For the second inequality, we use that $\mu_X + \lambda \sigma_X < \mu_Y -
\lambda \sigma_Y$. Therefore, $\Prob[Y \leq \mu_X + \lambda \sigma_X] \leq
\Prob[Y < \mu_Y - \lambda \sigma_Y]$ and $\Prob[X \geq \mu_Y - \lambda
\sigma_Y] \leq \Prob[X > \mu_X + \lambda \sigma_X]$.
For the last inequality, we apply Proposition~\ref{pr:chernoff_variance}
with $\lambda = 2\sqrt{\ln(\frac{2}{1-\eps})}$ and obtain that $\Prob[X > \mu_X + \lambda
\sigma_X] < (1-\eps)/2$ and that $\Prob[Y < \mu_Y - \lambda \sigma_Y] <
(1-\eps)/2$.
\qed

\subsection{The Proof of Lemma~\ref{l:lb}: Technical Details}
\label{s:app:lb_lemma}

We use the notation introduced in the proof sketch, in the main part. For
simplicity, we let $X(j)$ (resp. $Y(j)$) denote of the sum of $s$ Bernoulli
random variables with expectation $p_j^k = (1-a_j/\nu^{4j})^k$ (resp. $q_j^k =
(1-b_j/\nu^{4j})^k$), i.e., $X(j)$ and $Y(j)$ are the $k$-th PBD powers of the
Bernoulli trials in group $j$ of $\bm{p}$ and $\bm{q}$.

We first observe that for all $j < i$ and for all $a \in [\nu]$,
\begin{align*}
  s \left(1 - \frac{a}{\nu^{4j}}\right)^{\nu^{4i-2}}
  & \leq s\,e^{-a \nu^{4i-2} / \nu^{4j}} \\
  & = s\,e^{-a \nu^{4(i-j)-2}} \leq s/n^{a \ln n}\,,
\end{align*}
where for the last inequality, we use that $i - j \geq 1$ and that $\nu = \ln n$.
Therefore, since $a_j, b_j \geq 1$ and since $m < n$,
\[ \sum_{j < i} \Var[Y(j)] \leq \sum_{j < i} \Exp[Y(j)] \leq s /n^{-1+\ln n}
  \mbox{\,\ and} \]
\[ \sum_{j < i} \Var[X(j)] \leq \sum_{j < i} \Exp[X(j)] \leq s /n^{-1+\ln
    n}\,.\]
Moreover, $\sum_{j < i} \left|\Exp[Y(j)] - \Exp[X(j)] \right| \leq s /
n^{-1+\ln n}$.

We also have that the difference between the mean values of $X(i)$ and $Y(i)$
is:
\begin{align*}
  \Exp[X(i)] - \Exp[Y(i)]
  & = s \left(1-\frac{a_i}{\nu^{4i}}\right)^{\nu^{4i-2}} - s
  \left(1-\frac{a_i}{\nu^{4i}}-\frac{x_i}{\nu^{4i}}\right)^{\nu^{4i-2}} \\
  &\geq s \left( \nu^{4i-2} \frac{x_i}{\nu^{4i}}
    \left(1-\frac{a_i}{\nu^{4i}}\right)^{\nu^{4i-3}} -
    \left(\frac{\nu^{4i-2} x_i  }{2\nu^{4i}}\right)^2
    \left(1-\frac{a_i}{\nu^{4i}}\right)^{\nu^{4i-4}}\right)\\
  &\geq s \left( \frac{x_i}{2\nu^{2}} -
    \frac{1}{2}\left(\frac{x_i }{2\nu^{2}}\right)^2\right) \geq \frac{s x_i}{4
    \nu^{2}}
\end{align*}
The first inequality holds because for any $i$, $1 \geq \frac{a_i}{\nu^{2i}}+
\frac{x_i}{\nu^2}$. Therefore,
\[ \left(1-\frac{a_i}{\nu^{4i}}-\frac{x_i}{\nu^{4i}}\right)^{\nu^{4i-2}} \leq
  \left(1-\frac{a_i}{\nu^{4i}}\right)^{\nu^{4i-2}} -
  \nu^{4i-2} \frac{x_i}{\nu^{4i}}
  \left(1-\frac{a_i}{\nu^{4i}}\right)^{\nu^{4i-3}} +
  \left(\frac{\nu^{4i-2} x_i}{2\nu^{4i}}\right)^2
  \left(1-\frac{a_i}{\nu^{4i}}\right)^{\nu^{4i-4}}
\]
For the second inequality, we use that
$(1-\frac{a_i}{\nu^{4i}})^{\nu^{4i-3}} \geq 1 - \frac{a_i}{\nu^{3}} \geq 1/2$
and that
$(1-\frac{a_i}{\nu^{4i}})^{\nu^{4i-4}} \geq 1 - \frac{a_i}{\nu^{4}} \geq 1/2$.
For the last inequality, we use that $\frac{x_i}{2\nu^2} \leq 1/2$.

As for the variance of $X(i)$ and $Y(i)$, since $p_i^k > q_i^k > 1/2$ (assuming
that $n$ is sufficiently large),
\begin{align*}
  \Var[X(i)] \leq \Var[Y(i)]
  & = s \left(1-\frac{b_i}{\nu^{4i}}\right)^{\nu^{4i-2}}
  \left(1 - \left(1-\frac{b_i}{\nu^{4i}}\right)^{\nu^{4i-2}}\right) \\
  & \leq s \left(1 - \left(1-\frac{b_i}{\nu^{4i}}\right)^{\nu^{4i-2}}\right)\\
  & \leq s \nu^{4i-2} \frac{b_i}{\nu^{4i}} \leq s / \nu\,,
\end{align*}
where for the last inequality we use that $b_i \leq \nu$.

Moreover, we let $x_j = b_j - a_j$, with $0 \leq |x_j| < \nu$, and observe that
for all $j > i$,
\begin{align*}
  \left| \Exp[X(j)] - \Exp[Y(j)] \right|
  & = \left| s \left(1-\frac{a_j}{\nu^{4j}}\right)^{\nu^{4i-2}} - s
    \left(1-\frac{a_j}{\nu^{4j}}-\frac{x_j}{\nu^{4j}}\right)^{\nu^{4i-2}}
  \right| \\
  &\leq s \nu^{4i-2} \frac{|x_j|}{\nu^{4j}} \\
  &\leq \frac{s}{\nu^{4(j-i)+1}}
\end{align*}
Therefore, $\sum_{j > i} \left| \Exp[X(j)] - \Exp[Y(j)] \right| \leq 2 s/\nu^5$.

As for the variance of $X(j)$ and $Y(j)$, we have that for all $a \in [\nu]$,
\begin{align*}
  s \left(1-\frac{a}{\nu^{4j}}\right)^{\nu^{4i-2}}
  \left(1 - \left(1-\frac{a}{\nu^{4j}}\right)^{\nu^{4i-2}}\right)
  & \leq s \left(1 - \left(1-\frac{a}{\nu^{4j}}\right)^{\nu^{4i-2}}\right)\\
  & \leq s \frac{a}{\nu^{4(j-i)+2}} \\
  & \leq s / \nu^{4(j-i)+1}
\end{align*}
where for the last inequality we use that $a \in \nu$.
Therefore, $\sum_{j > i} (\Var[X(j)] + \Var[Y(j)]) \leq 4 s / \nu^{5}$.

Putting everything together and assuming that $n$ is large enough, we obtain
that
\begin{equation}\label{eq:exp_diff}
  \left| \Exp[X] - \Exp[Y]\right| \geq
  - \left| \sum_{j \neq i} (\Exp[X(j)] - \Exp[Y(j)]) \right| + (\Exp[X(i)] -
  \Exp[Y(i)])
  \geq s\left(\frac{x_i}{4\nu^2} - \frac{3}{\nu^5}\right)
  \geq \frac{x_i s}{5\nu^2}
\end{equation}
The first inequality holds because for all numbers $c$, $d$, $|c+d| \geq c -
|d|$ (here, we use $c = \Exp[X(i)] - \Exp[Y(i)]$ and $d = \sum_{j \neq i}
(\Exp[X(j)] - \Exp[Y(j)])$\,). The second inequality holds because (i)
$\Exp[X(i)] - \Exp[Y(i)] \geq s x_i / (4 \nu^2)$, (ii) $\sum_{j < i}
\left|\Exp[Y(j)] - \Exp[X(j)] \right| \leq s/n^{-1+\ln n}$, and (iii) $\sum_{j
  > i} \left| \Exp[X(j)] - \Exp[Y(j)] \right| \leq 2 s/\nu^5$.

As for the variance of $X$ and $Y$, we have that
\begin{equation}\label{eq:var_bound}
  \Var[X] + \Var[Y] \leq \frac{4s}{\nu}
\end{equation}

If we assume that $n$ is sufficiently large, that $s \geq (\ln n)^4$ and that
$a_i > b_i$, and thus $x_i \geq 1$, we obtain that for any $\eps \in (0, 1/2]$,
\begin{align*}
  \left| \Exp[X] - \Exp[Y] \right| & \geq \frac{s}{5\nu^2}
  > 4 \sqrt{\ln(\frac{2}{1-\eps}) \frac{4s}{\nu}}\\
  & \geq 4 \sqrt{\ln(\frac{2}{1-\eps}) (\Var[X] + \Var[Y])} \\
  & \geq 2 \sqrt{\ln(\frac{2}{1-\eps})} \left(\sqrt{\Var[X]} +
    \sqrt{\Var[Y]}\right)
\end{align*}
Therefore, by Lemma~\ref{l:tvd_lb}, we obtain that $X$ and $Y$ are at distance
larger than $\eps$, for any $\eps \in (0, 1/2]$, a contradiction.
Hence, it must be $a_i = b_i$, which concludes the proof of the lemma. \qed

\section{Upper Bound for Learning PBD Powers: Separated Case}

\subsection{The proof of Lemma~\ref{l:large_diff}}
\label{app:proof_lemma_large_diff}


\noindent
To prove this lemma we will first show the following.

\begin{lemma}\label{l:mon}
  Let us consider the following function $h(x) = \left( 1 - \frac{\gamma}{cx}
\right)^{x-1}$, where $c \geq 2$ is a fixed constant, $\gamma \geq 1$ and $x
\geq \gamma$. If $\gamma \geq 2c$, then function $h$ is increasing in $[\gamma,
+\infty)$.
\end{lemma}

\begin{proof}
  Let us observe that $h'(x) = \left( 1 - \frac{\gamma}{cx} \right)^{x-1}
\left(\ln(1-\frac{\gamma}{cx}) + \frac{\frac{\gamma}{cx}}{1-\frac{\gamma}{cx}}
\cdot \frac{x-1}{x}\right)$. Because we have that $h(x) \geq 0$, to show the
claim it suffices to prove that
$$
 \ln\left(1-\frac{\gamma}{cx}\right) +
\frac{\frac{\gamma}{cx}}{1-\frac{\gamma}{cx}} \cdot \frac{x-1}{x} \geq 0
$$ when $x \geq \gamma \geq 2c$.

This inequality can be rewritten to
$$
   \ln\left(1-\frac{\gamma}{cx}\right)^{-1} \leq
\frac{\frac{\gamma}{cx}}{1-\frac{\gamma}{cx}} \cdot   \frac{x-1}{x},
$$ which after variable change $z = \frac{\gamma}{cx} \in (0, \frac{1}{c}]$ is
equivalent to
$$
   \ln \frac{1}{1-z} \geq \frac{z}{1-z} \cdot \left(1 - \frac{c}{\gamma}
z\right).
$$ We now use the Taylor's series expansion $\ln \frac{1}{1-z} =
\sum_{i=1}^{\infty} \frac{z^i}{i}$, which leads to the following inequality
$$
   \frac{1-z}{z} \left(z + \frac{z^2}{2} + \frac{z^3}{3} + \frac{z^4}{4} +
\ldots \right) \leq 1 - \frac{c}{\gamma} z,
$$ or equivalently
$$
    \left(\frac{c}{\gamma} - \frac{1}{2}\right) z + \left(\frac{1}{3} -
\frac{1}{2}\right) z^2 + \left(\frac{1}{4} - \frac{1}{3}\right) z^3 +
\left(\frac{1}{4} - \frac{1}{5}\right) z^4 + \ldots \leq 0.
$$ Note that our operations on these infinite series are legitimate because
they converge: for instance taking series
$\sum_{i=1}^{\infty}\left(\frac{1}{i+1} - \frac{1}{i}\right) z^i =
\sum_{i=1}^{\infty} a_i$ and using the d'Alembert's ratio test we see that
$\lim_{i \rightarrow \infty} \left|\frac{a_{i+1}}{a_i}\right| = \lim_{i
\rightarrow \infty} \frac{i}{i+2} z = z < 1$, so the series converges.

Let us now observe that because $a_i < 0$ for all $i$ and $\frac{c}{\gamma} -
\frac{1}{2} \leq 0$, the inequality $(\frac{c}{\gamma} - \frac{1}{2}) z +
\sum_{i=2}^{\infty}(\frac{1}{i+1} - \frac{1}{i}) z^i \leq 0$ is true for all $z
\in (0,\frac{1}{c}]$.
\end{proof}

Now we are ready to prove Lemma \ref{l:large_diff}.

\begin{proof}
  Let $f(x) = (1-\frac{x}{c \ell_i})^{\ell_i}$, then $f'(x) = - \frac{1}{c}
(1-\frac{x}{c \ell_i})^{\ell_i - 1}$, and so by the Mean Value Theorem there
exists $\gamma_i \in (\alpha_i, \beta_i)$ such that  $(1-\frac{\alpha_i}{c
\ell_i})^{\ell_i} - (1-\frac{\beta_i}{c \ell_i})^{\ell_i} = f'(\gamma)
(\alpha_i - \beta_i)$. Using this, and that fact that $\beta_i - \alpha_i \geq
1$, the claimed inequality is implied if  the following inequality
$$
    \left(1-\frac{\gamma_i}{c \ell_i}\right)^{\ell_i - 1} > \frac{4c
\varepsilon}{\sqrt{n/s}}
$$ holds for all $n \geq e^{2c}$. Now because $\gamma_i \leq \sqrt{\ln(n)}$,
the last inequality is implied by
$$
    \left(1-\frac{\sqrt{\ln(n)}}{c \ell_i}\right)^{\ell_i - 1} > \frac{4c
\varepsilon}{\sqrt{n/s}},
$$ and by $\ell_i \geq \ln(n)$ and by Lemma \ref{l:mon} this inequality is
implied by
$$
    \left(1-\frac{\sqrt{\ln(n)}}{c \ln(n)}\right)^{\ln(n) - 1} > \frac{4c
\varepsilon}{\sqrt{n/s}}.
$$ The last inequality is equivalent to
$$
   \left(1 - \frac{1}{c \sqrt{\ln(n)}}\right)^{\ln(n)} > \frac{4c
\varepsilon}{\sqrt{n/s}} \left(1 - \frac{1}{c \sqrt{\ln(n)}}\right),
$$ which by using a known inequality $(1+y/m)^m \geq e^y (1-y^2/m)$, see Fact
\ref{fact:exp_ineq_1}, is implied by
$$
   \left(\frac{1}{e}\right)^{ \frac{ \sqrt{\ln(n)} }{c} } > \frac{4c
\varepsilon}{\sqrt{n/s}} \left(1 - \frac{1}{c
\sqrt{\ln(n)}}\right)/\left(1-\frac{1}{c^2}\right).
$$ We now observe that $\left(1 - \frac{1}{c
\sqrt{\ln(n)}}\right)/\left(1-\frac{1}{c^2}\right)  <  3/2$ and so the last
inequality is implied by
$$
    \left(\frac{1}{e}\right)^{\frac{\sqrt{\ln(n)}}{c} } \geq \frac{6c
\varepsilon}{\sqrt{n/s}}
\,\,\, \equiv \,\,\, \frac{\ln(n)}{2} \geq \frac{\sqrt{\ln(n)}}{c} +
\frac{\ln(s)}{2} + \ln(6c\varepsilon).
$$ The last inequality can easily be checked to hold when $n \geq e^{2c}$.
\end{proof}

\subsection{Proof of Theorem \ref{th:separated_upper}}
\label{app:proof_theorem_separated_upper}
\begin{proof}
  The main idea of the proof follows that of Algorithm \ref{alg:PBD_class}.
That is as it learns $p_i$'s starting from the largest and proceeding towards
the smallest, the proofs follows the same order.

\ignore{
Because $c$ is a fixed constant, if $n < e^{2c}$, then the problem is of
constant size and we can solve it in constant time by brute force. Assume
therefore that $n \geq e^{2c}$. Similarly we can assume that $n >
\frac{4}{(2-\sqrt{2})^2\varepsilon^2}$.
}

Recall that we assume that $n \geq e^{2c}$ and $n \geq
\frac{4}{(2-\sqrt{2})^2\varepsilon^2}$.

First we show how to exactly learn $p_0$. Observe that by the inequality from
Fact \ref{fact:exp_ineq_1}, we obtain
$$
p_0^{\ell_0} = \left(1 - \frac{\alpha_0/c}{(c \ln(n))^s/c}
\right)^{(c\ln(n))^s/c}
\leq \left(\frac{1}{e}\right)^{\alpha_0/c} \leq \left(\frac{1}{e}\right)^{1/c}
< 1.
$$ Similarly we see that $p_1^{\ell_0} \leq 1/n$ and $p_2^{\ell_0} \leq
(1/n)^{c \ln(n)}$, and in general $p_i^{\ell_0} \leq (1/n)^{(c \ln(n))^{i-1}}$
for $i = 1,2, \ldots$.

By this observation we can upper bound the mean of $X^{\ell_0}$ as follows
(note that $n_i = n_0 = n/s$ for all $i$):
\begin{equation}\label{equ-5}
    \Exp[X^{\ell_0}] = \sum_{i=0}^{s-1} n_i p_i^{\ell_0} = n_0 \sum_{i=0}^{s-1}
p_i^{\ell_0} = n_0 \cdot \left(p_0^{\ell_0} + \sum_{i=1}^{s-1}
p_i^{\ell_0}\right)  \leq
\end{equation}
$$
n_0 \cdot \left(p_0^{\ell_0} + \sum_{i=1}^{s-1} (1/n)^{(c \ln(n))^{i-1}}
\right) \leq
n_0 \cdot \left(p_0^{\ell_0} + \sum_{i=1}^{\infty} (1/n)^{i} \right) \leq
n_0 \cdot \left(p_0^{\ell_0} + 2/n \right),
$$ where the last estimate holds if $n \geq 2$.

Similarly we can bound the variance $\sigma^2_{\ell_0} = \Var[X^{\ell_0}]$:
$$
\Var[X^{\ell_0}] =  n_0 \sum_{i=0}^{s-1} p_i^{\ell_0} (1 - p_i^{\ell_0}) \leq
n_0 \cdot \left(p_0^{\ell_0} + 2/n \right) < n_0 \cdot \left(1 + 2/n \right)
\leq 2 n_0.
$$
We now draw $\bigOh{\frac{\log(s/\delta)}{\varepsilon^2 }}$ samples from
$X^{\ell_0}$ and obtain by Proposition \ref{pr:mu_sampling} the
estimate $\hat{\mu}_{\ell_0}$ of the mean $\mu_{\ell_0} = \Exp[X^{\ell_0}]$
such that
$$
   |  \mu_{\ell_0} - \hat{\mu}_{\ell_0} | \leq \varepsilon \sigma_{\ell_0} <
\varepsilon \sqrt{2 n_0},
$$ with probability at least $1-\delta/s$. This estimate, after letting
$\mu_{\ell_0} = n_0 \cdot \left(p_0^{\ell_0} + r_{\ell_0} \right)$, implies
\begin{equation}\label{equ-6}
   | p_0^{\ell_0} + r_{\ell_0} - \hat{\mu}_{\ell_0}/n_0 | < \varepsilon
\sqrt{2/n_0},
\end{equation} and $r_{\ell_0} \leq 2/n$ by (\ref{equ-5}). Let $\alpha_0 \in
\{1,2, \ldots, \lfloor\sqrt{\ln(n)}\rfloor\}$ be
such that $p_0^{\ell_0} = \left(1 - \frac{\alpha_0}{c \ell_0}\right)^{\ell_0}$,
then by (\ref{equ-6}) we obtain
\begin{equation}\label{equ-7}
  \left| \left(1 - \frac{\alpha_0}{c \ell_0}\right)^{\ell_0} + r_{\ell_0} -
\frac{\hat{\mu}_{\ell_0}}{n_0} \right| < \varepsilon \sqrt{2/n_0}.
\end{equation}

By $n \geq e^{2c}$, Lemma \ref{l:large_diff} implies that $\left(1 -
\frac{\beta_0 - 1}{c \ell_0}\right)^{\ell_0} - \left(1 - \frac{\beta_0}{c
\ell_0}\right)^{\ell_0} > \frac{4 \varepsilon}{\sqrt{n/s}}$,
for any $\beta_0 \in \{1,2, \ldots, \lfloor\sqrt{\ln(n)}\rfloor\}$. This,
together with (\ref{equ-7}), gives us that there exists the smallest $\beta_0
\in \{1,2, \ldots, \lfloor\sqrt{\ln(n)}\rfloor\}$ such that
$\left(1 - \frac{\beta_0}{c \ell_0}\right)^{\ell_0} \leq
\hat{\mu}_{\ell_0}/n_0$; thus, we have that
\begin{equation}\label{equ-6-1}
   \left(1 - \frac{\beta_0}{c \ell_0}\right)^{\ell_0} \leq
\frac{\hat{\mu}_{\ell_0}}{n_0} <
   \left(1 - \frac{\beta_0 - 1}{c \ell_0}\right)^{\ell_0},
\end{equation} and let indeed $\beta_0$ be such smallest number from $\{1,2,
\ldots, \lfloor\sqrt{\ln(n)}\rfloor\}$.

We will now prove that $\alpha_0 \in \{ \beta_0 - 1, \beta_0\}$ and we will
also show how to decide if in fact $\alpha_0 = \beta_0 - 1$ or $\alpha_0 =
\beta_0$, which means that we can learn the precise value of $p_0$ from
$\hat{\mu}_{\ell_0}/n_0$.

Suppose first that $\beta_0 \geq \alpha_0 + 2$; then by (\ref{equ-6-1}) we
obtain that
$$\left(1 - \frac{\alpha_0}{c \ell_0}\right)^{\ell_0} + r_{\ell_0} -
\frac{\hat{\mu}_{\ell_0}}{n_0} > \left(1 - \frac{\alpha_0}{c
\ell_0}\right)^{\ell_0} + r_{\ell_0} -
\left(1 - \frac{\beta_0 - 1}{c \ell_0}\right)^{\ell_0} \geq
$$
$$
\left(1 - \frac{\alpha_0}{c \ell_0}\right)^{\ell_0} + r_{\ell_0} - \left(1 -
\frac{\alpha_0 + 1}{c \ell_0}\right)^{\ell_0} >
\frac{4 \varepsilon}{\sqrt{n/s}} ,
$$ where the last inequality follows by Lemma \ref{l:large_diff} because $n
\geq e^{2c}$. But then this is in contradiction with (\ref{equ-7}); thus, we
must have that $\beta_0 \leq \alpha_0 + 1$.

Similarly, if  $\beta_0 \leq \alpha_0 - 1$, then by (\ref{equ-6-1})
$$
\left| \hat{\mu}_{\ell_0}/n_0 - p_0^{\ell_0} - r_{\ell_0} \right|  = \left|
\hat{\mu}_{\ell_0}/n_0  - \left(1 - \frac{\alpha_0}{c \ell_0}\right)^{\ell_0} -
r_{\ell_0} \right| \geq
$$
$$
\left| \left(1 - \frac{\alpha_0 - 1}{c \ell_0}\right)^{\ell_0}  - \left(1 -
\frac{\alpha_0}{c \ell_0}\right)^{\ell_0} - r_{\ell_0} \right| > \frac{4
\varepsilon}{\sqrt{n/s}} - 2/n ,
$$ where the last inequality follows by Lemma \ref{l:large_diff}. This
again is in contradiction with (\ref{equ-7}); thus, we have that $\beta_0 \geq
\alpha_0$. We have shown that $\beta_0 \in \{\alpha_0, \alpha_0 + 1\}$.

What remains to show is how to decide if $\alpha_0 = \beta_0 - 1$ or $\alpha_0
= \beta_0$.
By Lemma \ref{l:large_diff} the length of the interval $I_0 = \left[\left(1
- \frac{\beta_0}{c \ell_0}\right)^{\ell_0}, \left(1 - \frac{\beta_0 - 1}{c
\ell_0}\right)^{\ell_0}\right)$ in (\ref{equ-6-1}) containing number
$\hat{\mu}_{\ell_0}/n_0$ can be lower bounded as follows
$$
\left(1 - \frac{\beta_0 - 1}{c \ell_0}\right)^{\ell_0}  - \left(1 -
\frac{\beta_0}{c \ell_0}\right)^{\ell_0} > \frac{4 \varepsilon}{\sqrt{n/s}}.
$$

If we had $\alpha_0 = \beta_0$, then (\ref{equ-7}) implies that the distance
between numbers
$\left(1 - \frac{\beta_0}{c \ell_0}\right)^{\ell_0}$ and
$\hat{\mu}_{\ell_0}/n_0$ is at most  $r_{\ell_0} + \varepsilon \sqrt{2/n_0}
\leq 2/n + \varepsilon \sqrt{2/n_0}$, which is strictly less than half of the
length of interval $I_0$ by our assumption that $n >
\frac{4}{(2-\sqrt{2})^2\varepsilon^2}$. On the other hand if we had that
$\alpha_0 = \beta_0 - 1$, then (\ref{equ-7}) implies that the distance between
numbers
$\left(1 - \frac{\beta_0 - 1}{c \ell_0}\right)^{\ell_0}$ and
$\hat{\mu}_{\ell_0}/n_0$ is most
$\varepsilon \sqrt{2/n_0}$, which again is strictly less than half of the
length of interval $I_0$. We can therefore use this test to decide if $\alpha_0
= \beta_0 - 1$ or $\alpha_0 = \beta_0$.

This finishes the argument of how to exactly learn $p_0$. We will now assume
that for some $i \in \{1,2\ldots,s-1\}$, values of all $p_0, p_1, \ldots,
p_{i-1}$ are known exactly, and we will show how to learn exactly the value of
the next $p_i$.

The high-level argument will be quite similar to the case of learning $p_0$
because we now can assume that values of $p_0, p_1, \ldots, p_{i-1}$ are known.

By the inequality from Fact \ref{fact:exp_ineq_1}, we obtain
$$
p_i^{\ell_i} = \left(1 - \frac{\alpha_i/c}{(c \ln(n))^{s-i}/c}
\right)^{(c\ln(n))^{s-i}/c}
\leq \left(\frac{1}{e}\right)^{\alpha_i/c} \leq \left(\frac{1}{e}\right)^{1/c}
< 1.
$$ Similarly we see that $p_{i+1}^{\ell_i} \leq 1/n$ and $p_{i+2}^{\ell_i} \leq
(1/n)^{c \ln(n)}$, and in general $p_{i+j}^{\ell_i} \leq (1/n)^{(c
\ln(n))^{j-1}}$ for $j = 1,2, \ldots$.

We thus can upper bound the mean of $X^{\ell_i}$ as follows (note that $n_j =
n_0 = n/s$ for all $j$):
\begin{equation}\label{equ-8}
    \Exp[X^{\ell_i}] = \sum_{j=0}^{s-1} n_j p_j^{\ell_i} = n_0 \sum_{j=0}^{s-1}
p_j^{\ell_i} = n_0 \cdot \left(\sum_{j=0}^{i-1} p_j^{\ell_i} + p_i^{\ell_i} +
\sum_{j=i+1}^{s-1} p_j^{\ell_i}\right)  \leq
\end{equation}

$$
n_0 \cdot \left(p_0^{\ell_0} + \sum_{i=1}^{s-1} (1/n)^{(c \ln(n))^{i-1}} \right) \leq
n_0 \cdot \left(\sum_{j=0}^{i-1} p_j^{\ell_i} + p_i^{\ell_i} +
\sum_{j=1}^{\infty} (1/n)^{j} \right) \leq
n_0 \cdot \left(\sum_{j=0}^{i-1} p_j^{\ell_i} + p_i^{\ell_i} + 2/n \right),
$$ where the last estimate holds if $n \geq 2$.

We will now bound the variance $\sigma^2_{\ell_i} = \Var[X^{\ell_i}]$.
By the inequality $(1 - x/n)^n \geq 1 - x$, which holds for any $n \geq 1$ and
$x \leq n$, and by using that $\alpha_j \leq \ln(n)$, for all $j$, we obtain
$$
p_{i-1}^{\ell_i} = \left(1 - \frac{\frac{\alpha_{i-1}}{c^2 \ln(n)}}{ \frac{(c
\ln(n))^{s-i+1}}{c^2 \ln(n)} } \right)^{(c\ln(n))^{s-i}/c}
\geq 1 - \frac{\alpha_{i-1}}{c^2 \ln(n)} \geq 1 - \frac{1}{c^2 \sqrt{\ln(n)}} .
$$ Similarly we see that $p_{i-2}^{\ell_i} \geq 1 - \frac{1}{c^2 \sqrt{\ln(n)}
\cdot c \ln(n) } $ and $p_{i-3}^{\ell_i} \geq 1 - \frac{1}{c^2 \sqrt{\ln(n)}
\cdot (c \ln(n))^2 } $, and in general $p_{i-j}^{\ell_i} \geq 1 - \frac{1}{c^2
\sqrt{\ln(n)} \cdot (c \ln(n))^{j-1} } $ for $j = 1,2, \ldots$.
Thus we obtain
$$
\Var[X^{\ell_0}] =  n_0 \sum_{j=0}^{s-1} p_j^{\ell_i} (1 - p_j^{\ell_i}) \leq
n_0 \left(\sum_{j=0}^{i-1} (1 - p_j^{\ell_i}) + p_i^{\ell_i} +
\sum_{j=i+1}^{s-1} p_j^{\ell_i} \right)
\leq
$$
$$
n_0 \cdot \left(\left(\sum_{j=0}^{i-1} \frac{1}{c^2 \sqrt{\ln(n)} \cdot (c
\ln(n))^{j} }\right)  + p_i^{\ell_i} + 2/n \right) < n_0 \cdot \left(1/7 + 1 +
2/n \right) \leq 2 n_0,
$$ where we used that $n \geq e^{2c}$ and $c \geq 2$ to bound
$\left(\sum_{j=0}^{i-1} \frac{1}{c^2 \sqrt{\ln(n)} \cdot (c \ln(n))^{j}
}\right) \leq \left(\sum_{j=0}^{\infty} \frac{1}{c^2 \sqrt{\ln(n)} \cdot (c
\ln(n))^{j} }\right) \leq 1/7$.

We now draw $\bigOh{\frac{\log(s/\delta)}{\varepsilon^2}}$ samples from
$X^{\ell_i}$ and obtain by Proposition \ref{pr:mu_sampling} the
estimate $\hat{\mu}_{\ell_i}$ of the mean $\mu_{\ell_i} = \Exp[X^{\ell_i}]$
such that
$$
   |  \mu_{\ell_i} - \hat{\mu}_{\ell_i} | \leq \varepsilon \sigma_{\ell_i} <
\varepsilon \sqrt{2 n_0},
$$ with probability at least $1-\delta/s$. This estimate, after letting
$\mu_{\ell_i} = n_0 \cdot \left( \sum_{j=0}^{i-1} p_j^{\ell_i} + p_i^{\ell_i} +
r_{\ell_i} \right)$ implies
\begin{equation}\label{equ-9}
   \left| \sum_{j=0}^{i-1} p_j^{\ell_i} + p_i^{\ell_i} + r_{\ell_i} -
\hat{\mu}_{\ell_i}/n_0 \right| < \varepsilon \sqrt{2/n_0},
\end{equation} and $r_{\ell_i} \leq 2/n$ by (\ref{equ-8}). Let $\alpha_i \in
\{1,2, \ldots, \lfloor\sqrt{\ln(n)}\rfloor\}$ be
such that $p_i^{\ell_i} = \left(1 - \frac{\alpha_i}{c \ell_i}\right)^{\ell_i}$,
and let us also denote $\hat{\tau}_i = \hat{\mu}_{\ell_i}/n_0 -
\sum_{j=0}^{i-1} p_j^{\ell_i}$. Then by (\ref{equ-9}) we obtain
\begin{equation}\label{equ-11}
  \left| \left(1 - \frac{\alpha_i}{c \ell_i}\right)^{\ell_i} + r_{\ell_i} -
\hat{\tau}_i \right| < \varepsilon \sqrt{2/n_0}.
\end{equation}

Recall that the values $p_0, \ldots, p_{i-1}$ are known. Suppose next that we
find the smallest $\beta_i \in \{1,2, \ldots, \lfloor\sqrt{\ln(n)}\rfloor\}$
such that
$\left(1 - \frac{\beta_i}{c \ell_i}\right)^{\ell_i} \leq \hat{\tau}_i$ (such
$\beta_i$ exists by the same argument as that for $\beta_0$); thus, we have that
\begin{equation}\label{equ-10}
   \left(1 - \frac{\beta_i}{c \ell_i}\right)^{\ell_i} \leq  \hat{\tau}_i <
   \left(1 - \frac{\beta_i - 1}{c \ell_i}\right)^{\ell_i} .
\end{equation}

We will now prove that $\alpha_i \in \{ \beta_i - 1, \beta_i\}$ and we will
also show how to decide if $\alpha_i = \beta_i - 1$ or $\alpha_i = \beta_i$,
which will imply that the precise value of $p_i$ can be learned from
$\hat{\tau}_i$.

 Suppose first that $\beta_i \geq \alpha_i + 2$; then by (\ref{equ-10}) we
obtain that
$$\left(1 - \frac{\alpha_i}{c \ell_i}\right)^{\ell_i} + r_{\ell_i} -
\hat{\tau}_i > \left(1 - \frac{\alpha_i}{c \ell_i}\right)^{\ell_i} + r_{\ell_i}
-
\left(1 - \frac{\beta_i - 1}{c \ell_i}\right)^{\ell_i} \geq
$$
$$
\left(1 - \frac{\alpha_i}{c \ell_i}\right)^{\ell_i} + r_{\ell_i} - \left(1 -
\frac{\alpha_i + 1}{c \ell_i}\right)^{\ell_i} >
\frac{4 \varepsilon}{\sqrt{n/s}} ,
$$ where the last inequality follows by Lemma \ref{l:large_diff} because $n
\geq e^{2c}$. But then this is in contradiction with (\ref{equ-11}); thus, we
must have that $\beta_i \leq \alpha_i + 1$.

Similarly, if  $\beta_i \leq \alpha_i - 1$, then by (\ref{equ-10})
$$
\left| \hat{\tau}_i - p_i^{\ell_i} - r_{\ell_i} \right|  = \left| \hat{\tau}_i
- \left(1 - \frac{\alpha_i}{c \ell_i}\right)^{\ell_i} - r_{\ell_i} \right| \geq
$$
$$
\left| \left(1 - \frac{\alpha_i - 1}{c \ell_i}\right)^{\ell_i}  - \left(1 -
\frac{\alpha_i}{c \ell_i}\right)^{\ell_i} - r_{\ell_i} \right| > \frac{4
\varepsilon}{\sqrt{n/s}} - 2/n ,
$$ where the last inequality follows by Lemma \ref{l:large_diff}. This
again is in contradiction with (\ref{equ-11}); thus, we have that $\beta_i \geq
\alpha_i$. We have shown that $\beta_i \in \{\alpha_i, \alpha_i + 1\}$.

The next step is to decide if $\alpha_i = \beta_i - 1$ or $\alpha_i = \beta_i$.
By Lemma \ref{l:large_diff} the length of the interval $I_i = \left[\left(1
- \frac{\beta_i}{c \ell_i}\right)^{\ell_i}, \left(1 - \frac{\beta_i - 1}{c
\ell_i}\right)^{\ell_i}\right)$ in (\ref{equ-10}) containing number
$\hat{\tau}_i$ can be lower bounded as follows
$$
\left(1 - \frac{\beta_i - 1}{c \ell_i}\right)^{\ell_i}  - \left(1 -
\frac{\beta_i}{c \ell_i}\right)^{\ell_i} > \frac{4 \varepsilon}{\sqrt{n/s}}.
$$

If $\alpha_i = \beta_i$, then (\ref{equ-11}) implies that the distance between
numbers
$\left(1 - \frac{\beta_i}{c \ell_i}\right)^{\ell_i}$ and $\hat{\tau}_i$ is at
most  $r_{\ell_i} + \varepsilon \sqrt{2/n_0} \leq 2/n + \varepsilon
\sqrt{2/n_0}$, which is strictly less than half of the length of interval $I_i$
by our assumption that $n > \frac{4}{(2-\sqrt{2})^2\varepsilon^2}$. On the
other hand, if $\alpha_i = \beta_i - 1$, then (\ref{equ-11}) implies that the
distance between numbers
$\left(1 - \frac{\beta_i - 1}{c \ell_i}\right)^{\ell_i}$ and $\hat{\tau}_i$ is
most
$\varepsilon \sqrt{2/n_0}$, which again is strictly less than half of the
length of interval $I_i$. We can therefore use this test to decide if $\alpha_i
= \beta_i - 1$ or $\alpha_i = \beta_i$.

To finish the proof, observe that by the union bound all the
sampling estimates for the mean values $\hat{\mu_{\ell_i}}$
hold with probability at least $1-\delta$. Moreover, because
this sampling for each $i=0,1,\ldots,s-1$ draws $\bigOh{\frac{\log(s/\delta)}
  {\varepsilon^2}}$ samples from $X^{\ell_i}$, the total number of samples is
$\bigOh{\frac{s\log(s/\delta)}{\varepsilon^2} }$.
\end{proof}

\section{Upper Bound for Learning Binomial Powers}\label{s:app:binomial}
\subsection{Preliminaries} \label{Chernoff_facts}

We start by stating a useful variant of the standard Chernoff Bound.

\begin{proposition}[Chernoff Bound]\label{pr:chernoff_variance}
  Let $X = X_1+\cdots+X_n$, $X_i \in [0,1]$. Let $\mu = \Exp[X]$ and $\sigma^2 = \Var(X)$.
  Then, for all $\lambda \in (0, 2\sigma)$,
  $\Prob[X > \mu + \lambda \sigma] < e^{-\lambda^2/4}$ and
  $\Prob[X < \mu - \lambda \sigma] < e^{-\lambda^2/4}$\,.
\end{proposition}
\begin{proof}
  See, e.g., page 8 in the book \cite{DP09}.
\end{proof}

We now prove Fact \ref{fact:sampling} and Fact \ref{fact:binomial_UCL} on estimating the parameter
$p$ of a Binomial $B(n,p)$ using Chernoff's Bound. \\

\ignore{
\begin{fact}\label{fact:sampling}
  For any $\eps, \delta \in (0, 1/2)$, let
  $m = \lceil 4\ln(1/\delta)/(\eps^2 \psi^2) \rceil$
  and let $\p = (s_1+\cdots +s_m)/(m n)$, where $s_1, \ldots, s_m$ are $m$
  independent samples from a Binomial distribution $\B(n, p)$. Then,
  $\Prob[\p < p + \psi \err{n, p, \eps}] \geq 1-\delta,\
  \Prob[\p > p - \psi \err{n, p, \eps}] \geq 1-\delta$.
\end{fact}
}  

\noindent
{\bf Proof of Fact \ref{fact:sampling}.}

\begin{proof}
Let $X = \sum_{i=1}^m s_i/ n$. Then $s_i/n \in [0,1]$ and
$\Exp[X] = m p,\ \Var[X] = \frac{m}{n} p(1-p)$, since the samples
are i.d.d.
We show only that $\Pr[\p - p > \psi\ \err{n,p,\eps}] \leq \delta$
since the other case is similar.
From \ref{pr:chernoff_variance} we obtain with $t = m \psi\ \err{n,p,\eps} $
\begin{align*}
\Prob[\p - p > \psi \err{n,p,\eps} ]
=&
\Prob[\p - p > t/m]
\\
=&
  \Prob[X - \Exp[X] > \sqrt{m} \psi \eps \sqrt{\Var[X]} ]
  \\
  \leq&
  \exp\lp( - m \eps^2 \psi^2 / 4 \rp)
  \\
  \leq&
  \delta,
\end{align*}
where, for the last inequality, we use that
$m = \lceil 4\ln(1/\delta)/(\eps^2 \psi^2) \rceil$.
\qed
\end{proof}

\ignore{
\begin{fact}\label{fact:binomial_UCL}
  Let $B(n,p)$ be a Binomial distribution. Let
  $k = \lceil \ln(4/\delta)/ \ln(2) \rceil,\
  m = \lceil 4\ln(\lceil 2k/\delta \rceil)/(\eps^2 \psi^2) \rceil$.
  For each $i \in [k]$ let $w_i = \sum_{i=1}^m s_i/(nm)$, where
  $s_1,\ldots,s_m$ are $m$ independent samples from $B(n,p)$.
  Moreover, let
  $\q_1 = \min_{1 \leq i \leq k} w_i$,
  $\q_2 = \max_{1 \leq i \leq k} w_i$. Then
  \[
    \vProb{p - \psi \err{n,p,\eps} < \q_1 < p} \geq 1-\delta,\
    \vProb{p < \q_2 < p + \psi \err{n,p,\eps}} \geq 1-\delta.
  \]
  The overall number of samples to obtain each of $\q_1, \q_2$
  is $k m = O\lp(\ln(1/\delta)^2/(\eps^2 \psi^2)\rp)$.
\end{fact}
} 

\noindent
{\bf Proof of Fact \ref{fact:binomial_UCL}.}

\begin{proof}
We only prove that $\Prob\lp[ p<\ \q_2\ < p + \err{n,p,\eps} \rp] \geq 1 -\delta$
since the proof for $\q_1 = \min_{1 \leq i \leq k} w_i$ is essentially the same.
\begin{align*}
  \Prob &\lp[\lp(\max_i w_i < p \rp)\ \bigcup\
  \lp(\max_i w_i > p + \psi \err{n,p,\eps}\rp) \rp] \\
  \leq& \Prob \lp[\max_i w_i < p\rp] + \Prob\lp[\max_i w_i > p + \psi \err{n,p,\eps}\rp] \\
  = &\Prob \lp[\bigcap_{i=1}^k (w_i < p)\rp] +
  \Prob\lp[\bigcup_{i=1}^k \lp(w_i > p + \psi \err{n,p,\eps}\rp) \rp] \\
  \leq& \lp(\frac12 \rp)^k + k u\\
  \leq& \delta,
\end{align*}
where the last inequality follows from $k = \lceil \ln(2/\delta)/ \ln(2) \rceil$
and by choosing $u \leq \delta/(2k)$.
From Fact~\ref{fact:sampling} we have that
\[
  m = \lceil 4\ln(1/u)/(\eps^2 \psi^2) \rceil =
  \lp \lceil 4
  \frac{\ln
    \lp(\frac{2 \lp\lceil \ln(2/\delta)/ \ln(2) \rp\rceil}{\delta} \rp)
  } {\eps^2 \psi^2} \rp\rceil = O\lp(\frac{\ln(1/\delta)}{\eps^2 \psi^2}\rp)
\]
is sufficient to ensure that
$\vProb{w_i < p + \psi \err{n,p,\eps}} \leq \delta/(2k)$.
\qed
\end{proof}

\subsection{The Case where $p \in [\eps^2/n^d,\ 1- \eps^2/n^d]$.} \label{s:app:binomial:hugep}

We generalise Algorithm~\ref{alg:binomial} and its analysis to the case where
the value of $p$ is very close to $1$ or $0$ and lies in
$[\eps^2/n^d, 1-\eps^2/n^d]$, for some fixed constant integer $d \in \nats_{+}$.
Hence, we cover all values of $p$ that can be represented by $O(\log n)$ bits.

This will lead to the following theorem.

\begin{theorem}\label{t:hugep}
  Let $\eps \in (0, 1/6)$ and $d \in \nats_{+}$ be fixed constants,
  and let $n \in \nats$, $n \geq 5$.
  For any  $p \in [\eps^2/n^d, 1 - \eps^2/n^d]$, an extension of
  Algorithm~\ref{alg:binomial} uses $O(\log(d)\log(\log(d)/\delta)/\eps^{2})$
  samples and outputs $t, \hat{a} \in(0,+\infty)$, $\hat{q_1}, \hat{q_2} \in
(0,1)$ such
  that $\dtv{B(n,\q_2^l)}{B(n,p^{l t \hat{a}})} \leq O(\eps)$ for
  $l \in(0,1)$ and $\dtv{B(n, \q_1^l)}{B(n,p^{l t \hat{a}})}$
  for $l \in (1,+\infty)$ with probability at least $1-\delta$.
\end{theorem}

\begin{proof}
We will first describe the extension of
Algorithm~\ref{alg:binomial}. Notice that, in this case
we only need to find $t \in (0, +\infty)$ such
that $p^t \in [\eps^2/n, 1-\eps^2/n]$. Then, we simply call
Algorithm~\ref{alg:binomial} using $B(n,p^t)$ as the
"first" power to obtain $\q_1,\ \q_2,\ \hat{a}$ such that
$\dtv{B(n,\q_2^l)}{B(n,p^{l t \hat{a}})} \leq O(\eps)$ for
$l \in(0,1)$ and $\dtv{B(n, \q_1^l)}{B(n,p^{l t \hat{a}})}$
for $l \in (1,+\infty)$. To find $t$ we first sample from
$B(n,p)$ and using Fact~\ref{fact:binomial_UCL} we obtain
$\q_{1,1},\ \q_{1,2}$ such that
$\vProb{p - \err{n,p,\eps} < \q_{1,1} < p < \q_{1,2} < p + \err{n,p,\eps}} \geq 1- \delta$.
We have the following cases:
\begin{itemize}
  \item $\q_{1,1}>\eps^2/n$ and $\q_{1,2} < 1- \eps^2/n$. In this case we can use
    directly Algorithm~\ref{alg:binomial}.
  \item $\q_{1,1}<\eps^2/n$.
    Let $I_1 = \{1/(i\ln n)  : i\in \{2, \ldots, d\} \}$.
    Using Fact~\ref{fact:binomial_UCL} draw $O(d \ln^2(d/\delta) / \eps^2)$
    samples from the powers $B(n,p^l)$, $l \in I_1$, and obtain the set
    of approximations $Q_1 = \{ \q_{i,1} : i \in I_1\}$
    such that with probability $1-\delta/2$ all $\q_{i,1} \in Q_1$
    satisfy the bounds $p^i - \err{n,p^i,\eps} < \q_{i,1} < p^i$.
    We first prove that there exists an element $t$ of $I_1$ such
    that $\q_{t, 1} \geq \eps^2/n$. It suffices to show that such a $t$
    exists when $p = \eps^2/n^d$. Then
    $p^{1/(d \ln n) } = \eps^{2/(d \ln n)} /\me$ and
    $\q_{1/(d \ln n)} \geq p^{1/{(d \ln n)}} - \frac{\eps}{2 \sqrt{n} }
    \geq \eps^2/n$, for all $n \geq 7$.\\
    Let $t$ be the largest element of $I_1$ such that $\q_{t, 1} \geq \eps^2/n$.
    Then $p^t > \eps^2/n$ since $\q_{t,1} < p^t$. Moreover, $p^t < 1- \eps^2/n$.
    To show that write $t = 1/(\rho \ln n) $ for some $\rho \geq 2$ and
    $t' = 1/((\rho-1) \ln n) $.
    Then, $p^{t'} \leq \q_{t',1} + \err{n,p^{t'},\eps} \leq \eps^2/n + \eps/(2\sqrt{n}) \leq
    \eps/\sqrt{n}$. Thus, $p^t = p^{1/(\rho \log n) } =
    p^{\frac{1}{(\rho-1) \log n } \frac{\rho-1}{\rho}} = (p^{t'})^{\frac{\rho-1}{\rho}}
    \leq (\eps/\sqrt{n})^{\frac{\rho-1}{\rho}}= \frac{\eps} {\sqrt{n}}
    \lp( \frac {\sqrt{n}} {\eps} \rp)^{1/\rho} \leq \frac{\sqrt{\eps}}{n^{1/4}}
    \leq 1-\eps^2/n$, where the last inequality holds for $n\geq 2,\ \eps < 1/2$.
  \item $\q_{1,2}>1-\eps^2/n$.
    Consider now the set $I_2 = \{n^{i/3} : i \in \{0\} \cup [3d] \}$.
    Using Fact~\ref{fact:binomial_UCL} draw $O(d\ \ln^2(d/\delta)/\eps^2)$ samples
    and obtain the set of approximations $Q_2 = \{ \q_{i,2} : i \in I_2 \}$
    such that with probability $1-\delta/2$ all $\q_{i,2} \in Q_2$
    satisfy the bounds $p^i < \q_{i,2} < p^i + \err{n,p^i,\eps}$.
    As we did in the previous case we first prove that
    there exists a $t \in I_2$ such that $\q_{t,2} \leq 1-\eps^2/n$.
    It suffices to prove it for $p= 1-\eps^2/n^d$.
    Take $t= n^d$. Then
    $p^t = (1-\eps^2/n^d)^{n^d} \leq \me^{-\eps^2} \leq 1-\eps^2/2 \leq 1- \eps^2/n$
    for $\eps<0.85,\ n\geq 2$.

    Starting from $1$ find the smallest element $t$ of $I_2$ such that
    $\q_{t,2} < 1 -\eps^2/n$. We argue that $\eps^2/n < p^t < 1-\eps^2/n$.
    Obviously $p^t < 1- \eps^2/n$ since $p^t < \q_{t,2}$.
    To prove the other inequality, write $t =n^{\rho/3}$ and $t'= n^{(\rho-1)/3}$ for some
    $\rho \in [3d]$. We have that $\q_{t',2} \geq 1 - \eps^2/n$ and therefore
    $p^{t'} \geq 1-\eps^2/n -\err{n,p,\eps} \geq 1- \eps^2/n - \eps/(2\sqrt{n})
    \geq 1-\eps/\sqrt{n}$, for $n \geq 4$. Thus, $p^t = p^{n^{\rho/3}} = p^{n^{(\rho-1)/3 + 1/3}} =
    (p^{t'})^{n^{1/3}} \geq (1- \eps/\sqrt{n})^{n^{1/3}} \geq \me^{-2 \eps n^{1/3}/n^{1/2}}
    = \me^{-2 \eps/n^{1/6}} \geq \eps^2/n$, where for the second inequality we
    used $1- x \geq \me^{-2x}$ for $x \in [0, 0.75]$ and the
    last inequality holds for $\eps \leq 1/\me,\ n \geq 1$.
\end{itemize}
It is easy to see that we can improve the sampling complexity by doing
binary search on $d$, which means for each $\rho$ tested, the algorithm
chooses $O(\log(\log(d)/\delta)/ \eps^2)$ independent samples,
which leads to $O(\log(d)\log(\log(d)/\delta)/ \eps^2)$ total
number of samples. \qed
\end{proof}

\section{Lower Bound for Learning Binomial Powers}
\label{s:app:binomial_lower}

\subsection{Notation}
We denote by $\Normal{\mu}{\sigma}$ the Normal distribution
with mean $\mu$ and variance $\sigma^2$. The density function
of $\Normal{\mu}{\sigma}$ is $f(x) = \NormalPDF{\mu}{\sigma}{x}$.
We denote by $\erf{x}$ the Gauss error function, namely
$\erf{x} = \frac{2}{\sqrt{\pi}} \int_{0}^{x} \me^{-t^2}\, \mrm{d} t$,
by $\erfc{x}$ the complementary error function, $\erfc{x} = 1-\erf{x} =
\frac{2}{\sqrt{\pi}} \int_{x}^{+\infty} \me^{-t^2}\, \mrm{d}t$.

\subsection{Preliminaries}
We prove the following proposition which provides an
exact expression for the KL-Divergence of two Binomial distributions.
\begin{proposition}[Binomial KL-Divergence]\label{pr:binomial_kl_div}
  Let $X=B(n,p),\ Y=B(n,q)$ be two Binomial distributions.
  Then
  \[
    D_{kl}(X \| Y) = -n\lp((1-p)\log\lp(\frac{1-q}{1-p} \rp)
    +p\log\lp(\frac{q}{p}\rp) \rp)
  \]
\end{proposition}
\begin{proof}
  We have
  \begin{align*}
    \dkl{X}{Y}
    &= \sum_{k=0}^n X(k) \ln\frac{X(k)}{Y(k)}
    \\
    &= \sum_{k=0}^n \binom{n}{k} p^k (1-p)^{n-k}
    \ln\lp(\frac{p^k (1-p)^{n-k}}
    {q^k (1-q)^{n-k}}
    \rp)
    \\
    &= \sum_{k=0}^n \binom{n}{k} p^k (1-p)^{n-k}
    \lp(k \ln\lp(\frac p q \rp) +
    (n-k) \ln\lp(\frac {1-p}{1-q}\rp)
    \rp)
    \\
    &= \ln\lp(\frac p q \rp) \sum_{k=0}^n k \binom{n}{k} p^k (1-p)^{n-k} +
    \ln\lp(\frac {1-p}{1-q} \rp) \sum_{k=0}^n (n-k) \binom{n}{k} p^k (1-p)^{n-k}
    \\
    &= n p\, \ln\lp(\frac p q \rp) +
    n(1-p)\, \ln\lp(\frac {1-p}{1-q} \rp).
  \end{align*}
  \qed
\end{proof}
The following simple proposition formalizes an intuition that
when the distance of the parameters $p,q$ of two Binomial distributions
is large, then the Kullback-Leibler divergence of the these
distributions is large.
\begin{proposition}\label{pr:binomial_kl_monotonicity}
  Let $X\sim B(n,p),\ Y \sim B(n,q)$. Then
  $\dkl{X}{Y}$ and $\dkl{Y}{X}$ are both increasing functions of $|p-q|$.
\end{proposition}
\begin{proof}
  Write $q = p+x$. We start with $\dkl{B(n,p)}{B(n,p+x)}$.
  The derivative of $h(x)\coloneqq \dkl{B(n,p)}{B(n,p+x)}$ with respect to $x$ is
  \[
    h'(x) = \frac{n (1-p)}{-p-x+1}-\frac{n p}{p+x}.
  \]
  It's easy to see that $h'(x) > 0$ for $x \in (0,1-p)$ and $h'(x) < 0$ for
  $x \in(-p,0)$. Therefore $h(x)$ is minimized at $x=0$.
  Similarly if we let $g(x) = \dkl{B(n,p+x)}{B(n,p)}$ we have
  \[
    g'(x) = n \ln \left(\frac{p+x}{p}\right)-n
    \ln \left(\frac{-p-x+1}{1-p}\right).
  \]
  Again we have $g'(x) < 0$ for $x \in (-p, 0)$ and $g'(x) >0$
  for $x \in (0, 1-p)$. Thus, $g(x)$ is minimized at $x=0$.
  \qed
\end{proof}

\subsection{Discretized Normal Approximation}
We first start with some basic results about continuous Normal distributions.
Chu \cite{JTChu55} proved the following inequality for the Normal Integral
\begin{proposition}[Chu's Inequalities]\label{pr:chu_inequalities}
  For any $x \geq 0$:
  $$ \sqrt{1-\me^{-a x^2}} \leq \erf{x} \leq \sqrt{1-\me^{-b x^2}},
  $$
  where $a=1$ and $b = 4/\pi$.
\end{proposition}

The folowing Corollary of \ref{pr:chu_inequalities} provides a slightly
weaker lower bound for $\erf{x}$.
\begin{corollary}\label{co:erf_lower_bound}
  If $0 \leq x \leq 1$, then we have that $\erf{x} \geq x/c$,
  where $c$ is any fixed constant such that $c \geq \sqrt{\me/(\me -1)}$.
\end{corollary}

\begin{proof}
Using lower bound of the inequality of Proposition~\ref{pr:chu_inequalities}
with $a=1$ we want to prove that
$$
  \frac{x}{c} \leq \sqrt{1 - \frac{1}{\me^{x^2}}},
  $$ for any $x \in [0,1]$. This inequality is equivalent to
  $$
  f(x) := \me^{x^2}(c^2 - x^2) - c^2 \geq 0.
  $$ We have that $f'(x) = 2x \me^{x^2} (c^2 - 1 - x^2)$, therefore we see that
  $f'(x) \geq 0$ if $x \leq \sqrt{c^2-1}$ and  $f'(x) \leq 0$ if
  $x \geq \sqrt{c^2-1}$.
  This means that function $f$ has a maximum at $x_0 = \sqrt{c^2-1}$
  and thus its smallest value in $[0,1]$ is $\min\{f(0), f(1)\} = \{0, (\me-1)c^2 - \me\}$.
  Now we demand that $(\me-1)c^2 - \me \geq 0$, which leads
  to $c \geq \sqrt{\me/(\me -1)}$ and under this condition
  $f(x) \geq 0$ for all $x \in [0,1]$.
  \qed
\end{proof}
To bound $\erfc{z}$ we shall use Komatsu's inequality stated e.g. as
Problem 1, page 17 in \cite{ito_diffusion_1996}.  See \cite{yang_2015} for
more such results.
\begin{proposition}[Komatsu's Inequalities]\label{pr:komatsu}
  For all $a\geq 0$ it holds
  \[
    \frac{\me^{-a^2/2}}{2 \sqrt{a^2 + 4} +a} \leq \int_{a}^{+\infty}
    \me^{-t^2/2} \mrm{d}t \leq
    \frac{\me^{-a^2/2}}{2 \sqrt{a^2 + 2} +a}
  \]
\end{proposition}

When two continuous Normal distributions have the same standard deviation
a simple argument gives an exact expression for their total variation distance.
\begin{proposition}\label{pr:equal_variance_normal_tvd}
  Let $X \sim \Normal{\mu_1}{\sigma}$, $Y \sim \Normal{\mu_2}{\sigma}$.
  Then
  \[
    \dtv{X}{Y} = \erf{\frac{|\mu_1 - \mu_2|}{2 \sqrt{2} \sigma}}
  \]
\end{proposition}
\begin{proof}
  Let $f_X$ resp. $f_Y$ be the density functions for $X$ resp. $Y$.
  We can assume that $\mu_1 < \mu_2$ since the proof for the other case is
  essentially the same. Then
  \begin{align*}
    f_X(x) \geq f_Y(x)
    \Leftrightarrow
    \NormalPDF{\mu_1}{\sigma}{x} \geq \NormalPDF{\mu_2}{\sigma}{x}
    \Leftrightarrow
    |x-\mu_1| \leq |x-\mu_2|
    \Leftrightarrow
    x \leq \frac{\mu_1 + \mu_2}{2}
  \end{align*}
  Therefore,
  \[\dtv{X}{Y} =
    \int_{-\infty}^{\frac{\mu_1+\mu_2}{2}} \lp(f_X(x) - f_Y(x)\rp)\ \mrm{d}x =
    \erf{\frac{\mu_2 - \mu_1}{2 \sqrt{2} \sigma}}
  \]
  \qed
\end{proof}
When the variances of two Normal distributions differ the
following Proposition from \cite{DDS13}
provides an upper bound for their total variation distance.
\begin{proposition}[Proposition B.4 from \cite{DDS13}]\label{pr:normal_tvd_upper}
  Let $\mu_1, \mu_2 \in \reals$ and $0< \sigma_1\leq \sigma_2$. Then
  \[
    \dtv{\Normal{\mu_1}{\sigma_1}}{\Normal{\mu_2}{\sigma_2}} \leq
    \frac12 \lp( \frac{|\mu_1 - \mu_2|}{\sigma_1}
    + \frac{\sigma_2^2 - \sigma_1^2}{\sigma_1^2} \rp)
  \]
\end{proposition}

Let $X \sim \Normal{\mu}{\sigma}$. We denote by $\DNormal{\mu}{\sigma}$
the discretized Normal distribution, namely if $X_d \sim \DNormal{\mu}{\sigma}$
then $X_d$ is a discrete random variable with mass function
\[
  \vProb{X_d = k} = \vProb{ k - \frac12 < X \leq k + \frac12},
\] where $k$ is any integer.

The following recent result of Chen and Leong \cite{CGS2010}
(Theorem 7.1) shows that a continuity corrected discretized Normal distribution
approximates very well a PBD provided that the variance of the PBD is not very small.
\begin{lemma}[Theorem 7.1 from \cite{CGS2010}]\label{l:normal_approximation}
  Let $X$ be a PBD and let $\mu = \Exp[X]$, $\sigma^2 = \vVar{X}$.
  Let $Y \sim \DNormal{\mu}{\sigma}$. Then
  \[\dtv{X}{Y} \leq \frac{7.6}{\sigma}.\]
\end{lemma}
The next Lemma shows that $2$ discretized Normal distributions are close
if and only if the corresponding continuous Normals are close. The needed
condition for this to hold is that the variance of the $2$ Normals is not
too small.
\begin{lemma}[Discrete-Continuous Error]\label{l:discrete_continuous_error}
  Let $X = \Normal{\mu_1}{\sigma_1}$, $Y = \Normal{\mu_2}{\sigma_2}$  be two Normal
  distributions such that $\dtv{X}{Y} \geq \eps$, where $\eps >0$.
  Let $X_d \sim \DNormal{\mu_1}{\sigma_1}$, $Y_d \sim \DNormal{\mu_2}{\sigma_2}$.
  Then
  \begin{equation}\label{eq:discrete_continuous_error}
    \eps - \lp(\ell + m + u\rp) \leq \dtv{X_d}{Y_d} \leq \eps
  \end{equation}
  where
  \begin{align*}
    \ell =
    \frac14
    \lp(\erfc{ \frac{\mu_1}{\sqrt{2} \sigma_1}} +
    \erfc{ \frac{\mu_2}{\sqrt{2} \sigma_2}} \rp)
    &\qquad
    u =
    \frac14
    \lp(\erfc{ \frac{n-\mu_1}{\sqrt{2} \sigma_1}} +
    \erfc{\frac{n-\mu_2}{\sqrt{2} \sigma_2}} \rp)
    \\
    m = \frac{1}{2} \bigg( \erf{\frac{1}{\sqrt{2} \sigma_1}} &+
    \erf{\frac{1}{\sqrt{2} \sigma_2}} \bigg)
  \end{align*}
\end{lemma}
\begin{proof}
  Let $f_X$, resp. $f_Y$ be the density function of $X$, resp. $Y$.
  We have that
  $\int_{-\infty}^{+\infty} |f_X(x) - f_Y(x)| \mrm{d}x = 2 \dtv{X}{Y} = 2 \eps$.
  Since the density of a Normal distribution with mean $\mu$ is increasing
  in $(-\infty, \mu]$ and decreasing in $[\mu, +\infty)$ we have that
  there exist at most $2$ points $r_1,\ r_2$ where the sign of the
  difference $d(x) \coloneqq f_X(x) - f_Y(x)$ changes.
  Without loss of generality
  we assume that $d(x)$ is positive in
  $(-\infty, r_1)$ and $(r_2,+\infty)$ and negative in $[r_1, r_2]$.
  Now assume that $k_1 = \lfloor r_1 \rfloor$ and $k_2 = \lceil r_2 \rceil$.
  Now we can lower bound the total variation distance of $X_d$, $Y_d$
  \begin{align}\label{eq:discrete_tvd}
    2 \dtv{X_d}{Y_d} &=
    \sum_{k=0}^n
    \lp|\vProb{ k - \frac12 < X_d \leq k + \frac12} -
    \vProb{ k - \frac12 < Y_d \leq k + \frac12} \rp|
    \nonumber \\
    &= \sum_{k=0}^n
    \lp|\int_{k-1/2}^{k+1/2} f_X(x) - f_Y(x)\ \mrm{d}x \rp|
    \\
    &\geq
    \sum_{k=0}^{k_1}
    \int_{k-1/2}^{k+1/2} d(x)\ \mrm{d}x +
    \sum_{k=k_1+1}^{k_2}
    \int_{k-1/2}^{k+1/2} -d(x)\ \mrm{d}x +
    \sum_{k=k_2+1}^{n}
    \int_{k-1/2}^{k+1/2} d(x)\ \mrm{d}x
    \nonumber \\
    &\geq
    \int_{0}^{k_1} d(x)\ \mrm{d}x +
    \int_{k_1+1}^{k_2-1} -d(x)\ \mrm{d}x +
    \int_{k_2}^{n} d(x)\ \mrm{d}x.
    \nonumber
  \end{align}
  Since $\dtv{X}{Y} = \int_{-\infty}^{\infty} |f_X(X) - f_Y(x) |\ \mrm{d}x$ we
  need to upper bound the \enquote{missing} integrals in the above
  expression. We have
  \begin{align*}
    \int_{-\infty}^{0} |d(x)| \mrm{d}x \leq&
    \int_{-\infty}^{0} \lp( f_X(x) + f_Y(x) \rp) \mrm{d}x =
    \frac12
    \lp(
    \erfc{ \frac{\mu_1}{\sqrt{2} \sigma_1}} +
    \erfc{ \frac{\mu_2}{\sqrt{2} \sigma_2}}
    \rp)
  \end{align*}
  Similarly,
  \begin{align*}
    \int_{n}^{+\infty} |d(x)|\ \mrm{d}x \leq&
    \frac12
    \lp(\erfc{ \frac{n-\mu_1}{\sqrt{2} \sigma_1}} +
    \erfc{ \frac{n-\mu_2}{\sqrt{2} \sigma_2}} \rp)
  \end{align*}.
  Moreover,
  \begin{align*}
    \int_{k_1}^{k_1+1} |d(x)|\ \mrm{d}x +
    \int_{k_2-1}^{k_2} |d(x)|\ \mrm{d}x \leq&
    \int_{\mu_1 - 1}^{\mu_1 +1} f_X(x)  \mrm{d}x +
    \int_{\mu_2 - 1}^{\mu_2 +1} f_Y(x) \mrm{d}x
    \\
    \leq&
    \erf{\frac{1}{\sqrt{2} \sigma_1}} +
    \erf{\frac{1}{\sqrt{2} \sigma_2}}
  \end{align*}
  To prove the upper bound of inequality (\ref{eq:discrete_continuous_error})
  notice that using (\ref{eq:discrete_tvd}) we have
  \begin{align*}
  2 \dtv{X_d}{Y_d} &=
  \sum_{k=0}^n
    \lp|\int_{k-1/2}^{k+1/2} f_X(x) - f_Y(x)\ \mrm{d}x \rp|
    \\
    &\leq \sum_{k=0}^n
    \int_{k-1/2}^{k+1/2} \lp|f_X(x) - f_Y(x)\ \rp| \mrm{d}x
    \\
    &\leq
    \int_{-\infty}^{+\infty} \lp|f_X(x) - f_Y(x)\ \rp|\mrm{d}x
    \\
    &= 2\eps
  \end{align*}
  \qed
\end{proof}

\subsection{The proof of Theorem~\ref{th:binomial_lower_bound}}

We remark that finding a family of sequences satisfying
$\rho(\theta(\mcal{P}), \theta(\mcal{Q})) = \bigOmega{\delta}$
instead of $2\delta$ changes the lower bound only by a constant factor.
Thus, to simplify our analysis, we shall not compute the constants for the
lower bound. Note that these will be absolute constants, independent
from any parameter, like $\eps$, $\delta$, etc., in our setting.

We restate explicitly our family of Binomial power sequences for the
shake of completeness.
Let $\delta = \Theta(1/\sqrt{n\, N})$.
Let $p_1 = 1/2$, $p_2 = 1/2 + \delta/4$, $p_3 = 1/2 + \delta/2$.
Let $P_{1,1} = B(n,p_1)$, $P_{2,1} = B(n,p_2), P_{1,3} = B(n,p_3)
$ be three Binomial distributions with corresponding power sequences
$\mcal{P}_1 = (B(n,p_1^i))_{i \in (1, +\infty)}$,
$\mcal{P}_2 = (B(n,p_2^i))_{i \in (1,+\infty)}$,
$\mcal{P}_3 = (B(n,p_3^i))_{i \in (1,+\infty)}$.

For the total variation distance of any of the above pairs
$i,j \in \{1,2,3\}$, $i \not= j$, we have
$
\dtv{\mcal{P}_i}{\mcal{P}_j} = \bigOmega{1/\sqrt{N}}.
$
Without loss of generality we prove that
$\dtv{\mcal{P}_1}{\mcal{P}_2} = \Omega(1/\sqrt{N})$.
From the definition of total variation distance for
sequences of distributions we see that to lower bound
the metric $\rho$ we just need to prove that the total
variation distance of $P_{1,1}, P_{2,1}$ is $\Omega(1/\sqrt{N})$,
namely we need to consider only the first power of the sequences.

Let $\mu_1 = \Exp[P_{1,1}]$, $\mu_2 = \Exp[P_{2,1}]$,
$\sigma_1^2 = \Var[P_{1,1}]$, $\sigma_2^2 = \Var[P_{2,1}]$.

We first use Lemma \ref{l:normal_approximation} to approximate
$P_{1,1}, P_{2,1}$ with discretized Normal distributions
$\DNormal{\mu_1}{\sigma_1}$, $\DNormal{\mu_2}{\sigma_2}$.
Since $\sigma_1, \sigma_2$ are both $O(\sqrt{n})$ the error
of the two discretized Normal approximations is $O(1/\sqrt{n})$.
From Proposition~\ref{pr:normal_tvd_upper} we obtain that we can
approximate $\Normal{\mu_2}{\sigma_2}$ using a Normal with the same
mean but with variance $\sigma_1^2$.
Applying Proposition~\ref{pr:normal_tvd_upper} yields
\[
  \dtv{\Normal{\mu_2}{\sigma_2}}{\Normal{\mu_2}{\sigma_1}}
  \leq \frac12 \frac{\sigma_1^2 - \sigma_2^2}{\sigma_2^2}
  = \frac12 \frac{n/4 - n(1/4 - \delta^2/16)}{n(1/4 - \delta^2/16)}
  = \bigOh{\delta^2}
  = \bigOh{1/n}
\]
Consider now the pair of continuous Normals $\Normal{\mu_1}{\sigma_1} $
$\Normal{\mu_2}{\sigma_1}$. Using Proposition~\ref{pr:equal_variance_normal_tvd}
we have that
\[
  \dtv{\Normal{\mu_1}{\sigma_1}}{\Normal{\mu_2}{\sigma_1}}
  = \erf{\frac{n \delta}{4 \sqrt{2} \sqrt{n}}}
  = \erf{\frac{\sqrt{n}\delta}{4 \sqrt{2}}}
  = \erf{\frac{1}{4 \sqrt{2}\sqrt{N}}}
  \geq \frac{1}{9 \sqrt{N}},
\] by Corollary~\ref{co:erf_lower_bound}.
Therefore, by the triangle inequality, we have
\[
  \dtv{\Normal{\mu_1}{\sigma_1}}{\Normal{\mu_2}{\sigma_2}}
  \geq \frac{1}{9 \sqrt{N}} - \bigOh{\frac{1}{n}}
\]
Applying Lemma \ref{l:discrete_continuous_error} when
$\sigma_1, \sigma_2$ are $\bigOh{\sqrt{n}}$ yields
$
\ell + m + u = \bigOh{\frac{1}{\sqrt{n}}},
$
since from Komatsu's inequalities (Proposition~\ref{pr:komatsu})
we have $\erfc{\sqrt{n}} = \Theta\lp(\frac{\me^{-n}}{\sqrt{n}}\rp)$
and from Proposition~\ref{pr:chu_inequalities} we have that
$\erf{1/\sqrt{n}} = \Theta\lp(\frac{1}{\sqrt{n}}\rp)$.
Therefore
\[
  \dtv{\DNormal{\mu_1}{\sigma_1}}{\DNormal{\mu_2}{\sigma_2}}
  \geq \frac{1}{9 \sqrt{N}} - \bigOh{\frac{1}{\sqrt{n}}}
\]
Overall, using triangle inequality and the above bounds we have
that
\[
  \dtv{P_{1,1}}{P_{2,1}}
  \geq \frac{1}{9 \sqrt{N}} - \bigOh{\frac{1}{\sqrt{n}}}
\]

We continue with proving an upper bound for the Kullback-Leibler divergence
between \emph{all} powers, namely the $\sup_{a\in \nats} \dkl{P_{1,a}}{P_{3,a}}$.
To apply Theorem~\ref{l:minimax_fano_lower_bound} it suffices to show that
the following holds
$\sup_{i,j \in[3],\, a\in \nats} \dkl{P_{i,a}}{P_{j,a}} = \bigOh{1/N}$.
From Proposition~\ref{pr:binomial_kl_monotonicity} it is clear that we only
need to bound the Kullback-Leibler distance for the most distant $p_i$'s,
namely the distances $\sup_{a \in \nats} \dkl{P_{1,a}}{P_{3,a}}$,
$\sup_{a \in \nats} \dkl{P_{3,a}}{P_{1,a}}$. We remark that is easy to verify
that $\dkl{P_{1,a}}{P_{3,a}} \approx \dkl{P_{3,a}}{P_{1,a}}$
for all $a \in \nats$ and therefore we will bound $\dkl{P_{1,a}}{P_{3,a}}$.

Applying Proposition~\ref{pr:binomial_kl_div} for $P_{1,a}$ and $P_{3,a}$ gives
\begin{align*}
  \dkl{P_{1,a}}{P_{3,a}} =
  2^{-a} n \ln
  \left(
    2^{-a} \left(\frac{1}{\frac{\delta}{2}+\frac{1}{2}}\right)^a
  \right)+
  \left(1-2^{-a}\right) n \ln
  \left(
    \frac{1-2^{-a}}{1-\left(\frac{\delta}{2}+\frac{1}{2}\right)^a}
  \right)
\end{align*}
Let $f(\delta) = \dkl{P_{1, a}}{P_{3,a}}$ defined by the above expression.
Taylor Expanding $f(\delta)$ around $0$ gives
$f(\delta) = 0 + R_1(z)$ for a $z \in [-\delta, \delta]$.
To bound the error of the Taylor approximation
we bound the derivative $f''(\delta)$
\begin{align*}
  f''(\delta) &=
  \frac {a n \left(\left(2^a-1\right) a (\delta +1)^a-\left(2^a-(\delta +1)^a\right)
      \left((\delta +1)^a-1\right)\right)}
  {(\delta +1)^2 \left(2^a-(\delta +1)^a\right)^2}
  \\
  &\leq
  \frac {a n \left(\left(2^a-1\right) a (\delta +1)^a \right)}
  {(\delta +1)^2 \left(2^a-(\delta +1)^a\right)^2}
  \\
  &\leq
  \frac {n a^2 (2^a-1) (3/2)^a}
  {\left(2^a-(3/2)^a \right)^2}
  \\
  &\leq
  \frac {n a^2 3^a}
  {\left( 2^a /4 \right)^2} =
  16 n\, a^2 (3/4)^a \leq 105 n
\end{align*}
Therefore, $\dkl{P_{1,i}}{P_{3,i}} \leq |R_1(z)| \leq 105 n\,
\delta^2 \leq 105/N$ for all $i \in \nats$.
\qed

\end{document}